\documentclass[runningheads]{llncs}
\usepackage[letterpaper,margin=1.3in]{geometry}
\usepackage{geometry}
\usepackage[english]{babel}
\usepackage[utf8]{inputenc}
\usepackage[]{amsmath}
\usepackage{amsfonts}
\usepackage{amssymb}
\usepackage{tikz}
\usepackage{hyperref}
\usepackage{caption}
\usepackage{xcolor}
\usepackage{float}
\usepackage{tikz,tikz-3dplot}
\usepackage{graphicx}
\usetikzlibrary{decorations.pathreplacing,shapes.geometric,arrows,positioning}
\usepackage[linesnumbered,ruled,vlined,resetcount]{algorithm2e}
\usepackage{calligra}
\usepackage{adjustbox}
\usepackage{animate}
\usepackage{breqn}
\usepackage{lipsum}
\usepackage{dsfont}
\usepackage[utf8]{inputenc}
\usepackage[T1]{fontenc}
\usepackage{float}
\usepackage{resizegather}
\usepackage[version=4]{mhchem}
\usepackage{stackrel,amssymb}
\usepackage{tikz}
\usetikzlibrary{arrows}
\usepackage{adjustbox}
\usepackage{calligra}
\usepackage[font=small,labelfont=bf]{caption}
\usepackage{mathtools}
\usepackage{tikz} 
\usepackage{tkz-graph}
\usepackage{svg}
\usepackage{yhmath}
\usepackage{cancel}
\usepackage{color}
\usepackage{siunitx}
\usepackage{array}
\usepackage{multirow}
\usepackage{amssymb}
\usepackage{gensymb}
\usetikzlibrary{fadings}
\usepackage{mathtools}
\usepackage{bm}
\usepackage{esvect}
\usepackage{apptools}
\newtheorem{teo}{Theorem}
\newtheorem{prop}{Proposition}
\makeatletter
\def\fourvdots{\vbox{\baselineskip1\p@ \lineskiplimit\z@
		\kern6\p@\hbox{.}\hbox{.}\hbox{.}\hbox{.}}}
\makeatother
\begin{document}
\title{On the effects of firing memory in the dynamics of conjunctive networks.}
%
%
\author{Eric Goles \inst{1} \and
	Pedro Montealegre \inst{1} \and
    Mart\'in R\'ios Wilson \inst{2} }
\authorrunning{Eric Goles, Pedro Montealegre and Mart\'in R\'ios Wilson.}
%
\institute{Universidad Adolfo Ibáñez, Facultad de Ingeniería y Ciencias, Santiago, Chile. \email{eric.chacc@uai.cl,p.montealegre@uai.cl  \and Universidad de Chile, Facultad de Ciencias Físicas y Matemáticas, Departamento de Ingeniería Matemática, Santiago, Chile. \email{mrios@dim.uchile.cl}}}
\maketitle              
\begin{abstract}
Boolean networks are one of the most studied discrete models in the context of the study of gene expression. In order to define the dynamics associated to a Boolean network, there are several \textit{update schemes} that range from parallel or \textit{synchronous} to \textit{asynchronous.} However, studying each possible dynamics defined by different update schemes might not be efficient. In this context, considering some type of temporal delay in the dynamics of Boolean networks emerges as an alternative approach. In this paper, we focus in studying the effect of a particular type of delay called \textit{firing memory} in the dynamics of Boolean networks. Particularly, we focus in symmetric (non-directed) conjunctive networks and we show that there exist examples that exhibit attractors of non-polynomial period. In addition, we study the prediction problem consisting in determinate if some vertex will eventually change its state, given an initial condition. We prove that this problem is \textbf{PSPACE}-complete.
\end{abstract}
\section{Introduction}
Boolean networks are one of the simplest and most studied discrete models in the context of the study of gene expression \cite{kauffman1969metabolic,kauffman1974large,thomas1973boolean}. A boolean network is defined by a boolean map that is usually represented as graph, called interaction graph, where the vertices or nodes represent  genes and the edges represent regulatory interactions. A gene in the network can be active or inactive and that is represented by a node in \textit{state} $1$ or $0$ respectively. This model was first introduced by Kauffman in the end of the 60's  \cite{kauffman1969metabolic} and it was thought as a generalization of the McCulloch and Pitts neural network model \cite{mcculloch1943logical}. The seminal papers by Kauffman and Thomas focused in studying the dynamical properties of random generated networks \cite{kauffman1969metabolic,kauffman1969homeostasis,kaufmann1971gene} as well as studying the structure involved in the regulatory circuitry \cite{thomas1973boolean,thomas1991regulatory}.   A boolean network naturally defines a discrete dynamical system by updating all the nodes of the network simultaneously, i.e. the consecutive states of the dynamics are given by iterations of the original boolean map. This update scheme is called \textit{parallel} or \textit{synchronous}. As the number of posible states is finite (it is given by the number of possible tuples with values $0$ and $1$ which is $2^n$) every initial state eventually exhibit periodic dynamics. We call the set of states that define these periodic dynamics an \textit{attractor}. If the attractor is one single state we call it a \textit{fixed point} and otherwhise we call it a \textit{limit cycle}. Though this model is fairly simple to study, it fails to reproduce gene expression data in a realistic way, mainly because of the synchronous update scheme. One straightforward approach to improve the model is adding asynchronicity to the dynamics by considering different update schemes \cite{aracena2009robustness,demongeot2008robustness,goles2008comparison,robert2012discrete}. Since some biologists agreed that some synchronicity is not completely unrealistic \cite{boettiger2009synchronous,wang2011spatial} updates schemes usually range from synchronous to sequential update schemes in which every node is updated according to a given partial order. A notable example of update schemes that are somewhere between the latter categories are \textit{block sequential} update schemes. In these update schemes, a partition in the node set is defined and nodes inside each set of the partition are updated in parallel while sets in the partition are updated sequentially. However, in order to define one of the latter update schemes, a partition and an partial order need to be chosen. These requirements introduce, in the biological networks modelling framework, several ways to model the dynamics of a fixed object of study. Although it is relevant and interesting from a mathematical or computational point of view to study the dynamics generated by every possible update scheme in the latter context, this exercise might turn to be rather impractical  \\[12pt]
An alternative approach to allow adding asynchronicity to the dynamics of a boolean network is based in the concept of delay that is generally defined as an internal clock, that could be independent from the original dynamics of the system, and that dictates its dynamical behaviour during a fixed time interval. This latter concept was first introduced by Thomas in \cite{thomas1991regulatory,thomas1995dynamical} and then studied in different frameworks such as in \cite{ahmad2008analysing,bernot2004application,fromentin2010hybrid,ren2008asymptotic,ribeiro2015learning}. Particularly, here we are interested in specific type of delay called \textit{firing memory}. It was based in the concept of memory and it was first introduced by Graudenzi and Serra under the name of \textit{gene protein Boolean networks} \cite{graudenzi2010new,graudenzi2011dynamical,graudenzi2011robustness} and they defined this delay inspired in the concept of decay of proteins. In \cite{goles2017firing}, Goles et al. introduced some modifications of the original model and presented it under the name of \textit{Boolean networks with firing memory.} A question that naturally arise in this context is what are the effects of firing memory in the dynamics of boolean networks. According to \cite{goles2017firing}, one of the first observations stated in the seminal papers, that was deduced through the analysis of numerical simulations, is that the more maximum time decay value (delay) the less the network admits asymptotic degrees of freedom. In order to survey this observation from a theoretical point of view, a straightforward methodology is to study a specific class of boolean networks preferably the one where the dynamics have been characterized. In fact, we are interested in the effect of \textit{firing memory} in the dynamics of threshold networks.  In these networks, the state of every node evolves accordingly to a \textit{threshold function} that depends on the state of certain variables represented as the neighbors of the node in underlying  interaction graph. In \cite{goles1985decreasing} Goles et. al characterized the dynamics of the latter network (without delay) and particularly they showed that attractors can only be limit cycles with period $2$ or fixed points. One of the simplest type of threshold networks are the \textit{disjunctive} and \textit{conjunctive} boolean networks in which the state of every node depends on an OR or an AND function of its neighbors respectively. In \cite{goles2017firing}, Goles et. al proved that 
disjunctive networks with firing memory only admit homogeneous fixed points as attractors. However, the effect of these type of delay in  the dynamics of conjunctive networks have not been described until now, perhaps surprisingly, conjunctive networks with firing memory can exhibit extremely different behaviour compared to the regular ones.  \\[12pt]
In the latter context, an interesting question is if firing memory is able to induce in the original boolean network  dynamics the capability of simulating other computation models such as boolean circuits, Turing machines, etc. This line of research led us to consider a natural problem that arise in the study of boolean network dynamics: the \textit{prediction problem.} This problem is defined in the following way: given an initial condition and an update scheme(in this case parallel scheme with firing memory), to predict the future states. To solve that problem, several strategies can be proposed from directly simulating the network to more elaborated strategies based on the topological or algebraical properties of the network. A measure of the efficiency of an strategy is given by the computational complexity of the problem.  Prediction problems have been broadly studied in threshold networks \cite{goles2014computational,goles2016pspace} and particularly, in disjunctive(conjunctive) networks, it is known that the problem is in the \textbf{P} class. \\[12pt]
In this paper, we focus in studying the dynamics of conjunctive networks with firing memory and we prove that, contrary to what might be assumed based on previous results for disjunctive networks, conjunctive networks with firing memory admit attractors of non polynomial period. Then, we study the prediction problem and we prove that it is \textbf{PSPACE}-complete. We achieve this by showing that conjunctive networks with firing memory are capable of simulating iterated boolean circuits. As a direct corollary of this result, we conclude that the latter boolean networks with firing memory are universal, in the sense that they are able to simulate an arbitrary given boolean function.
\section{Contributions and structure of the paper}
In this paper we show that, contrarily to what one may think, conjunctive networks with firing memory exhibit an extremely complex dynamical behavior. More precisely,  we show that \textsc{2-And-Prediction} is \textbf{PSPACE}-complete as a consequence of the capability of this rule to simulate iterated monotone boolean circuits.  As a corollary of the latter result, we show that conjunctive boolean networks with firing memory are a universal model in the sense that they are capable of simulating every boolean network automata. \\

The rest of the paper is organized as follows. In Section \ref{sec:prel} we give the main formal definitions and previous results. In Section \ref{sec:cycle} we show the gadgets that play an essential role in the proof of our main results and we use them to exhibit a conjunctive network with firing memory that admits attractors of non polynomial period.  In Section \ref{sec:PSPACE} we study the computational complexity of the \textsc{2-And-Prediction} problem and we give a complete proof of the main result in Theorem \ref{teo:PSPACE}. 
\section{Preliminaries}\label{sec:prel}
A boolean network is a map $F: \{0,1\}^n \to \{0,1\}^n$. Associated to this function, we define its interaction graph $\mathcal{G}(f) = (V,E)$  by $V = \{0,\hdots, 1\}$ and $ij \in E \iff F_j \text{ depends on the variable } x_i.$ $F$ defines a dynamical system $(X= \{0,1\}^n, F)$ in which the elements $x \in X$ are called \textit{states} or \textit{configurations.} and the transitions are given by the iterations of the map $F$, i.e, for every state $x \in X$ we define its next state by $x(1) = F(x)$ and in general we have that $x(t+1) = F(x(t))$ for every $t \in \mathbb{N}$. This type of dynamics is often called \textit{parallel} or \textit{synchronous} update scheme. In the next sections we will assume that boolean networks dynamics will be defined in this way.  Given an initial condition $x \in X$, we call its associated \textit{trajectory} to the infinite sequence $T(x) = (x,x(1),\hdots).$ As the number of possible states is finite ($2^n$), every trajectory is eventually periodic, i.e., there exists $p\geq0$ such that $x(t+p) = x(t)$. We say that a trajectory reaches a \textit{limit cycle} with period $p$ if the last property hold for that trajectory and $p$ is the minimum time in which the property is satisfied. A set of configurations in a limit cycle with period $p$ is called an \textit{attractor} with period $p$.  Particularly, when $p = 1$ we say that the attractor is a \textit{fixed point.}\\[12pt]
We are interested in some specific type of boolean networks called \textit{threshold} networks. A \textit{threshold} network is a boolean network in which given a matrix $A = (a_{ij})$ with integer entries and an integer vector $\Theta = (\theta_i)_i$  the function $F$ is defined by $$ F(x)_i  = \begin{cases}
1 & \text{ if } \sum \limits_{j=1}^{n} a_{ij} x_j - \theta_i \geq 0 \\
0 & \text{ otherwise } 
\end{cases}$$
One particular class of threshold networks are disjunctive and conjunctive networks.  Disjunctive networks are  defined by threshold $0$ in every coordinate function $F_i$, in other words,  are defined by an OR of certain variables i.e., $F(x)_i = F_i(x_{j_1}, \hdots, x_{j_k}) = \bigvee\limits_{i=1}^{k} x_{j_i}.$ On the other hand, a conjunctive network is a boolean network $F$ such that every local rule $F_i$ is given by an AND function of certain variables, i.e., $F(x)_i = F_i(x_{j_1}, \hdots, x_{j_k}) = \bigwedge\limits_{i=1}^{k} x_{j_i}.$ In this case, we have that $\theta_i = \delta_i$ for every $i$ where $\delta_i$ is the number of neighbors of $i$ in the interaction graph associated to $F$. These networks have been broadly studied in different frameworks \cite{goles2012disjunctive,jarrah2010dynamics,aracena2017fixed,gao2018controllability} mainly because its simplicity and its relevance in applications in modelling gene regulatory networks in which conjunctive functions describe common regulatory interactions  \cite{nguyen2006deciphering,gummow2006reciprocal}.\\[12pt]
As we referred in the introduction, the concept of delay in boolean networks has emerged as an alternative approach to introduce asynchronicity. In particular, we are interested in studying the effects of a type of delay called \textit{firing memory}.  We consider a boolean network $F$ and states $Y = \prod_{i = 1}^{n}\{0,1\} \times \{1,\hdots,dt_i\}$, $dt_i \geq 1$ for all $i$.  Given $y(0) \in X, \text{ } y(0)_i = (x(0)_i, \Delta(0)_i),$ we define the following dynamics:	
\begin{equation*}
x_i(t+1) = \left\{ \begin{array}{cc}
1 & \Delta_i(t+1) \geq 1, \\
F_i(x(t)) &  \Delta_i(t+1) = 0.\\
\end{array} \right.
\end{equation*}
\begin{equation*}
\Delta_i(t+1) = \left\{ \begin{array}{cc}
dt_i    & F_i(x(t)) = 1, \\
\Delta_i(t) -1 &  F_i(x(t)) = 0 \land \Delta_i(t) \geq 1, \\
0 & F_i(x(t)) = 0 \land \Delta_i(t) = 0. \\
\end{array} \right.
\end{equation*}	
 This local rule $(x_i(t),\Delta_i(t) \to  (x_i(t+1), \Delta_i(t+1))$ defines a global transition function $F^{dt}: Y \to Y$ that we call boolean network with \textit{firing memory.}  One useful notation introduced in \cite{goles2017firing}, is considering the states as the single delay value instead of a tuple. For example, the state $(1,2)$ that means state $1$ and delay $2$ is represented exclusively by $2$. In the next sections, we will be using this notation.  \\[12pt]
A natural question regarding the effects of firing memory in the dynamics of conjunctive networks is if this type of delay is able to give simulation capabilities to the network in the sense of allowing it to simulate other boolean networks of a different class or other computation models. In this context, we are interested in studying \textit{prediction problems.} A well studied topic in the context of the dynamics of boolean networks is to make predictions about the attractor associated to an specific trajectory defined by a initial condition $x$. There exists a very simple solution to this problem that is to simulate the network dynamics until the initial state reaches a limit cycle. However, a question that naturally arise is if there exist a more efficient solution, considering the fact that using the last strategy may take as many steps as there are possible states. These more efficient solutions would be based on algorithmic or algebraic properties of the network.  If $x \in \{0,1\}^n$ we introduce the complement of $x$ denoted by  $\bar{x}$ and defined by: $x_i = 1$ implies $\bar{x}_i = 0$ and $x_i = 0$ implies $\bar{x}_i = 0$. Given a maximum delay vector $dt$ we define the following decision problem:\\[6pt]    
 \fbox{\begin{minipage}{0.98\textwidth}
 $dt$-\textsc{and-prediction}: \\
Given a conjunctive network with firing memory $F$ with maximun delay vector $dt$,  $i \in \{1,\hdots,\}$ and a configuration $x \in \{0,1\}^n$, does there exist $y \in T(x)$ such that $y_i = \bar{x}_i$ ?
 \end{minipage}}\\[7pt]  
We remark that, because we are working with the AND rule, if $x_i = 0$ and $F(x)_i = 0$ then $x(t)_i = 0$ for all $t \geq 2$, so the $dt$-\textsc{and-prediction} problem can be solved simulating one step of the local rule. In this case and we can asume that $x_i = 1$. In addition, when we consider the maximum delay vector as uniform, i.e. the same maximum delay $dt_i = \tau$ in every $i$, we will refer to $dt$-\textsc{and-prediction} as $\tau$-\textsc{and-prediction}.  We are interested in studying the computational complexity of the previous problem. This concept is roughly defined as the amount of resources that are needed to find a solution, given as an expression of the input size. Classical theory defines the following main classes of complexity:  \textbf{P} is the class of problems solvable by a deterministic Turing machine in polynomial time and \textbf{PSPACE} is the class of problems solvable by a deterministic machine that uses polynomial space. Additionally, it is known that $ \textbf{P} \subset \textbf{PSPACE}.$ It is conjectured that these inclusions are strict, so there are problems in \textbf{PSPACE} that do not belong to {\bf P}. The problems in {\bf PSPACE} that are the most likely to not belong to {\bf P}  are the \textbf{PSPACE}-complete problems, which analogously \textbf{NP}-complete problems, are the ones such that any other problem in \textbf{PSPACE} can be reduced to them in polynomial time.  \\[7pt]
One very well known type of \textbf{PSPACE}-complete problem is related to the iterative evaluation of boolean circuits. A boolean circuit is a directed acyclic graph $C$ whose have three types of vertices: the ones with in-degree $0$ called \textit{inputs}, the ones with out-degree $0$ called \textit{outputs} and the rest of the vertices that have in and out neighbors called logical gates. These nodes are labelled by $\wedge, \vee, \neg$.  A boolean circuit simulates a boolean function in the obvious way, and because of that, usually a circuit with $n$ inputs and $m$ outputs is denoted  by $C: \{0,1\}^n \to \{0,1\}^m$. A circuit $C$ is monotone if there are no gates labelled by $\neg$. For each gate of a circuit, its \textit{layer} is the length of the shortest path from an input to the gate. A monotone circuit is alternating if for any path from an input to an output the gates on the path alternate between OR and AND gates. In addition, the inputs are connected to OR gates exclusively and outputs are OR gates. We define the following decision problem: Given a (monotone) boolean circuit $C:\{0,1\}^n \to \{0,1\}^n$, an input $x \in \{0,1\}^n$, and $i \in \{0,\hdots,n\}$ whether there exists a time $t \geq 1$ such that $C^t(x)_i = 1$. We call this problem \textsc{Iter-Circuit-Prediction} (respectively \textsc{Iter-Mon-Circuit-Prediction}).
\begin{proposition}
	\textsc{Iter-Mon-Circuit-Prediction} is \textbf{PSPACE}-complete even when restricted to alternating circuits of degree 4.
	\label{prop:PSPACECirc}
\end{proposition}
 \subsection{Previous results}
Threshold networks were vastly studied in \cite{chacc1980comportement,goles1985decreasing,hopfield1982neural,mortveit2007introduction} and Goles et. al \cite{goles1985decreasing} showed using a technique based on monotone energy operator that the \textit{synchronous} dynamics associated to these networks admits only attractors of bounded period (moreover, there are only attractors with period $2$ and fixed points) when the associated weight matrix is symmetric, i.e. when the underlying interaction graph is non-directed which is the case we are most interested on in this paper. \\[7pt]
A wide studied subclass of threshold networks are the conjunctive or disjuntive networks. These systems have a very important role in modelling of biological systems because mainly because of its simplicity and their straightforward way to describe common interactions between different variables. The dynamics of these type of networks was studied under different update schemes in \cite{goles2012disjunctive}. Our approach here is to continue the studies presented in the seminal paper of the \textit{firing memory}  \cite{goles2017firing} in which dynamics of disjunctive networks with firing memory were characterized. In the latter paper,  Goles et. al showed that disjuntive networks with firing memory and delay $dt_i \geq 2$ in at least one coordinate $i$ admit only homogeneous fixed points when the network is defined over a strongly connected directed graph.
\section{Conjunctive networks with firing memory admit non-polynomial cycles.}
\label{sec:cycle}
 One surprising observation about the effect of firing memory in the dynamics of conjunctive networks is that it allows the dynamics to have attractors with period $p\geq3$. We recall that these type of dynamics(without delay) have bounded cycles of maximum period $p=2$  \cite{goles1985decreasing}. A general method to construct  dynamics with a given period $p\geq 3$ is considering a conjunctive network defined by an interaction graph given by a complete graph $K_{p+1}$ and a firing memory $dt_i = p$ in every vertex.
 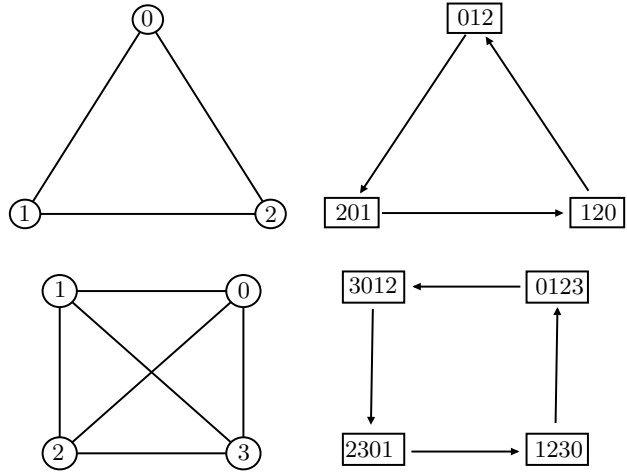
\begin{figure}[t]
 	\centering

 		\tikzset{every picture/.style={line width=0.75pt}} 
 		
 		\begin{tikzpicture}[x=0.75pt,y=0.75pt,yscale=-0.45,xscale=0.48]
 		
 		\draw   (154.5,21) -- (283.5,239) -- (25.5,239) -- cycle ;
 		\draw  [fill={rgb, 255:red, 255; green, 255; blue, 255 }  ,fill opacity=1 ] (267.5,239) .. controls (267.5,230.16) and (274.66,223) .. (283.5,223) .. controls (292.34,223) and (299.5,230.16) .. (299.5,239) .. controls (299.5,247.84) and (292.34,255) .. (283.5,255) .. controls (274.66,255) and (267.5,247.84) .. (267.5,239) -- cycle ;
 		\draw  [fill={rgb, 255:red, 255; green, 255; blue, 255 }  ,fill opacity=1 ] (9.5,239) .. controls (9.5,230.16) and (16.66,223) .. (25.5,223) .. controls (34.34,223) and (41.5,230.16) .. (41.5,239) .. controls (41.5,247.84) and (34.34,255) .. (25.5,255) .. controls (16.66,255) and (9.5,247.84) .. (9.5,239) -- cycle ;
 		\draw  [fill={rgb, 255:red, 255; green, 255; blue, 255 }  ,fill opacity=1 ] (138.5,21) .. controls (138.5,12.16) and (145.66,5) .. (154.5,5) .. controls (163.34,5) and (170.5,12.16) .. (170.5,21) .. controls (170.5,29.84) and (163.34,37) .. (154.5,37) .. controls (145.66,37) and (138.5,29.84) .. (138.5,21) -- cycle ;

 		\draw  [fill={rgb, 255:red, 255; green, 255; blue, 255 }  ,fill opacity=1 ] (469.5,3) -- (526,3) -- (526,36) -- (469.5,36) -- cycle ;
 		\draw  [fill={rgb, 255:red, 255; green, 255; blue, 255 }  ,fill opacity=1 ] (340.5,221) -- (397,221) -- (397,254) -- (340.5,254) -- cycle ;
 		\draw  [fill={rgb, 255:red, 255; green, 255; blue, 255 }  ,fill opacity=1 ] (598.5,221) -- (655,221) -- (655,254) -- (598.5,254) -- cycle ;
 		\draw    (488.5,38) -- (378.56,214.3) ;
 		\draw [shift={(377.5,216)}, rotate = 301.95] [fill={rgb, 255:red, 0; green, 0; blue, 0 }  ][line width=0.75]  [draw opacity=0] (8.93,-4.29) -- (0,0) -- (8.93,4.29) -- cycle    ;
 		
 		\draw    (589.5,238) -- (400.5,238) ;
 		
 		\draw [shift={(591.5,238)}, rotate = 180] [fill={rgb, 255:red, 0; green, 0; blue, 0 }  ][line width=0.75]  [draw opacity=0] (8.93,-4.29) -- (0,0) -- (8.93,4.29) -- cycle    ;
 		\draw    (617.5,213) -- (510.56,41.7) ;
 		\draw [shift={(509.5,40)}, rotate = 418.02] [fill={rgb, 255:red, 0; green, 0; blue, 0 }  ][line width=0.75]  [draw opacity=0] (8.93,-4.29) -- (0,0) -- (8.93,4.29) -- cycle    ;

 		\draw (154,20) node  [align=left] {$\displaystyle 0$};
 		\draw (25,238) node  [align=left] {$\displaystyle 1$};
 		\draw (284,239) node  [align=left] {$\displaystyle 2$};
 		\draw (499,18) node  [align=left] {$\displaystyle 012$};
 		\draw (371,237) node  [align=left] {$\displaystyle 201$};
 		\draw (628,237) node  [align=left] {$\displaystyle 120$};

 		\end{tikzpicture}\label{fig:sub1}
 		

	\vspace{0.5cm}

 		\tikzset{every picture/.style={line width=0.75pt}} 
 		
 		\begin{tikzpicture}[x=0.75pt,y=0.75pt,yscale=-0.45,xscale=0.48]
 		
 		\draw    (38.5,45) -- (231.5,227) ;

 		\draw    (38.5,227) -- (231.5,45) ;

 		\draw   (38.5,45) -- (231.5,45) -- (231.5,227) -- (38.5,227) -- cycle ;
 		\draw  [fill={rgb, 255:red, 255; green, 255; blue, 255 }  ,fill opacity=1 ] (213.5,45) .. controls (213.5,35.06) and (221.56,27) .. (231.5,27) .. controls (241.44,27) and (249.5,35.06) .. (249.5,45) .. controls (249.5,54.94) and (241.44,63) .. (231.5,63) .. controls (221.56,63) and (213.5,54.94) .. (213.5,45) -- cycle ;
 		\draw  [fill={rgb, 255:red, 255; green, 255; blue, 255 }  ,fill opacity=1 ] (213.5,227) .. controls (213.5,217.06) and (221.56,209) .. (231.5,209) .. controls (241.44,209) and (249.5,217.06) .. (249.5,227) .. controls (249.5,236.94) and (241.44,245) .. (231.5,245) .. controls (221.56,245) and (213.5,236.94) .. (213.5,227) -- cycle ;
 		\draw  [fill={rgb, 255:red, 255; green, 255; blue, 255 }  ,fill opacity=1 ] (20.5,227) .. controls (20.5,217.06) and (28.56,209) .. (38.5,209) .. controls (48.44,209) and (56.5,217.06) .. (56.5,227) .. controls (56.5,236.94) and (48.44,245) .. (38.5,245) .. controls (28.56,245) and (20.5,236.94) .. (20.5,227) -- cycle ;
 		\draw  [fill={rgb, 255:red, 255; green, 255; blue, 255 }  ,fill opacity=1 ] (20.5,45) .. controls (20.5,35.06) and (28.56,27) .. (38.5,27) .. controls (48.44,27) and (56.5,35.06) .. (56.5,45) .. controls (56.5,54.94) and (48.44,63) .. (38.5,63) .. controls (28.56,63) and (20.5,54.94) .. (20.5,45) -- cycle ;
 		
 		\draw  [fill={rgb, 255:red, 255; green, 255; blue, 255 }  ,fill opacity=1 ] (336,23) -- (400.5,23) -- (400.5,58) -- (336,58) -- cycle ;
 		\draw  [fill={rgb, 255:red, 255; green, 255; blue, 255 }  ,fill opacity=1 ] (336,205) -- (400.5,205) -- (400.5,240) -- (336,240) -- cycle ;
 		\draw  [fill={rgb, 255:red, 255; green, 255; blue, 255 }  ,fill opacity=1 ] (529,205) -- (593.5,205) -- (593.5,240) -- (529,240) -- cycle ;
 		\draw  [fill={rgb, 255:red, 255; green, 255; blue, 255 }  ,fill opacity=1 ] (529,23) -- (593.5,23) -- (593.5,58) -- (529,58) -- cycle ;
 		\draw    (411,40) -- (522.5,40) ;
 		
 		\draw [shift={(409,40)}, rotate = 0] [fill={rgb, 255:red, 0; green, 0; blue, 0 }  ][line width=0.75]  [draw opacity=0] (8.93,-4.29) -- (0,0) -- (8.93,4.29) -- cycle    ;
 		\draw    (365.53,193) -- (367.5,65) ;
 		
 		\draw [shift={(365.5,195)}, rotate = 270.88] [fill={rgb, 255:red, 0; green, 0; blue, 0 }  ][line width=0.75]  [draw opacity=0] (8.93,-4.29) -- (0,0) -- (8.93,4.29) -- cycle    ;
 		\draw    (559.5,193) -- (561.47,65) ;
 		\draw [shift={(561.5,63)}, rotate = 450.88] [fill={rgb, 255:red, 0; green, 0; blue, 0 }  ][line width=0.75]  [draw opacity=0] (8.93,-4.29) -- (0,0) -- (8.93,4.29) -- cycle    ;
 		
 		\draw    (408,225) -- (519.5,225) ;
 		\draw [shift={(521.5,225)}, rotate = 180] [fill={rgb, 255:red, 0; green, 0; blue, 0 }  ][line width=0.75]  [draw opacity=0] (8.93,-4.29) -- (0,0) -- (8.93,4.29) -- cycle    ;

 		\draw (39,43) node  [align=left] {$\displaystyle 1$};
 		\draw (232,44) node  [align=left] {$\displaystyle 0$};
 		\draw (232,226) node  [align=left] {$\displaystyle 3$};
 		\draw (37,226) node  [align=left] {$\displaystyle 2$};
 		\draw (564,40) node  [align=left] {$\displaystyle 0123$};
 		\draw (368,39) node  [align=left] {$\displaystyle 3012$};
 		\draw (364,222) node  [align=left] {$\displaystyle 2301$};
 		\draw (563,222) node  [align=left] {$\displaystyle 1230$};

 		\end{tikzpicture}
 		
 		\label{fig:sub2}
 	
 	\caption{Attractors with period $p= \tau +1$ for $\tau = 2$ and $\tau = 3$. a) Conjunctive network with firing memory and delay $dt_i = 2$ in every node that admits a limit cycle with period $3$. On the left hand side we show the interaction graph of the network and on the right hand side there is the transitions graph of the cycle. b) Conjunctive network with firing memory and delay $dt_i = 3$ in every node that admits a limit cycle with period $4$. On the left hand side we show the interaction graph of the network and on the right hand side there is the transitions graph of the cycle.}
 	\label{fig:example}
 \end{figure}
\begin{proposition}
	Let $\tau \geq 2$. There exists a conjunctive network with firing memory and delay $dt_i = \tau$ in every $i$ allowing attractors with period $p = \tau + 1$.
	\label{prop:cycles}
\end{proposition}
\begin{proof}
	Let's consider the conjunctive network given by the function $F(x)_i =  \bigwedge\limits_{j \neq i} ^{} x_{j},$ where $i \in \{0,\hdots, \tau \}$.  We have that its interaction graph is given by the complete graph with $\tau+1$ vertices $K_{\tau+1}$ (see Figure \ref{fig:example}). We are going to exhibit an attractor $X$ with period $p = \tau +1$. Let $x_0 = (0123\hdots \tau)$ be the initial condition. Observe that every $i \in V$ has initial delay $i$ and every node $i \neq 0$ is in state $1$. Because of how we defined $F$ we have that in the next state every node $i \neq 0$ will be set to $0$. However, the only node that will actually be in state $0$ in its next state is $i = 1$, because every node $ i \in V \setminus \{0,1\}$ has delay $\Delta(0)_i \geq 2$. On the other hand, we have that for $i = 0$ every of its neighbours is in state $1$ in the initial condition, so it will be updated as $\tau$ in the next iteration. Thus, $x_1 = (\tau0123\hdots \tau-1)$. Now, we have that every node except $i = 1$ is in state $1$ and every node $i \in i \in V \setminus \{1,2\}$ has delay $\Delta_i(1) \geq 2$ so using the same argument we used for deducing $x_1$ we have that $x_2 = (\tau-1 \tau 0 1 2 3 \hdots, \tau-2)$. Iterating this process $\tau+1$ times we verify we have a cycle with period $p = \tau + 1$:	
	\begin{equation*}
	X=\left  \{ \begin{array}{clc}
	x_0 & =  (0123\hdots \tau-1 \tau) & \\
	x_1 & = (\tau0123\hdots \tau-1) & \\
	x_2 & =  (\tau-1 \tau 0 1 2 3 \hdots, \tau-2) & \\
	\vdots & & \\
	x_{\tau} & = (123 \hdots \tau-1\tau 0) & \\
	x_0 & = (0123\hdots \tau-1\tau) & \\
	\end{array} \right.
	\end{equation*}
\end{proof}
As a consequence of the last proposition, we have a stronger result on the period of the attractors when we consider conjunctive networks with firing memory with different maximum delay values.
\begin{theorem}
	There exists a connected conjunctive network with firing memory (and not necessarily the same values for maximum delay) which admits attractors with non polynomial period. 
	\label{teononpolydif}
\end{theorem}
The main idea in this proof is to use the latter proposition to construct conjunctive networks with firing memory such that each of these networks admits attractors with prime period. Then, connecting this components as it is shown in Figure \ref{fig:nonpoly} we get a connected network that admits attractors with non polynomial period. This technique was first introduced in \cite{kiwi1994no} and a complete proof is available in the Appendix Section \ref{sec:nonpolyproof}.
 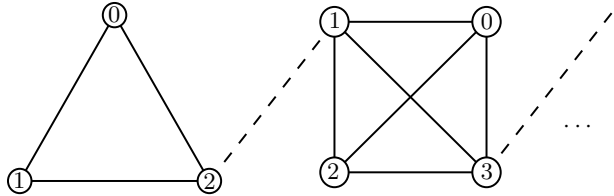
\begin{figure}
	\centering
	\tikzset{every picture/.style={line width=0.75pt}} 
	
	\begin{tikzpicture}[x=0.75pt,y=0.75pt,yscale=-0.5,xscale=0.5]
	
	\draw  [dash pattern={on 4.5pt off 4.5pt}]  (628.36,29.29) -- (502.19,189.97) ;

	\draw  [dash pattern={on 4.5pt off 4.5pt}]  (348.81,38.03) -- (222.64,198.71) ;

	\draw    (348.81,38.03) -- (502.19,189.97) ;

	\draw    (348.81,189.97) -- (502.19,38.03) ;

	\draw   (348.81,38.03) -- (502.19,38.03) -- (502.19,189.97) -- (348.81,189.97) -- cycle ;
	\draw  [fill={rgb, 255:red, 255; green, 255; blue, 255 }  ,fill opacity=1 ] (487.89,38.03) .. controls (487.89,29.73) and (494.29,23) .. (502.19,23) .. controls (510.1,23) and (516.5,29.73) .. (516.5,38.03) .. controls (516.5,46.33) and (510.1,53.06) .. (502.19,53.06) .. controls (494.29,53.06) and (487.89,46.33) .. (487.89,38.03) -- cycle ;
	\draw  [fill={rgb, 255:red, 255; green, 255; blue, 255 }  ,fill opacity=1 ] (487.89,189.97) .. controls (487.89,181.67) and (494.29,174.94) .. (502.19,174.94) .. controls (510.1,174.94) and (516.5,181.67) .. (516.5,189.97) .. controls (516.5,198.27) and (510.1,205) .. (502.19,205) .. controls (494.29,205) and (487.89,198.27) .. (487.89,189.97) -- cycle ;
	\draw  [fill={rgb, 255:red, 255; green, 255; blue, 255 }  ,fill opacity=1 ] (334.5,189.97) .. controls (334.5,181.67) and (340.9,174.94) .. (348.81,174.94) .. controls (356.71,174.94) and (363.11,181.67) .. (363.11,189.97) .. controls (363.11,198.27) and (356.71,205) .. (348.81,205) .. controls (340.9,205) and (334.5,198.27) .. (334.5,189.97) -- cycle ;
	\draw  [fill={rgb, 255:red, 255; green, 255; blue, 255 }  ,fill opacity=1 ] (334.5,38.03) .. controls (334.5,29.73) and (340.9,23) .. (348.81,23) .. controls (356.71,23) and (363.11,29.73) .. (363.11,38.03) .. controls (363.11,46.33) and (356.71,53.06) .. (348.81,53.06) .. controls (340.9,53.06) and (334.5,46.33) .. (334.5,38.03) -- cycle ;
	
	\draw   (127,31.29) -- (222.64,198.71) -- (31.36,198.71) -- cycle ;
	\draw  [fill={rgb, 255:red, 255; green, 255; blue, 255 }  ,fill opacity=1 ] (210.78,198.71) .. controls (210.78,191.93) and (216.09,186.42) .. (222.64,186.42) .. controls (229.19,186.42) and (234.5,191.93) .. (234.5,198.71) .. controls (234.5,205.5) and (229.19,211) .. (222.64,211) .. controls (216.09,211) and (210.78,205.5) .. (210.78,198.71) -- cycle ;
	\draw  [fill={rgb, 255:red, 255; green, 255; blue, 255 }  ,fill opacity=1 ] (19.5,198.71) .. controls (19.5,191.93) and (24.81,186.42) .. (31.36,186.42) .. controls (37.91,186.42) and (43.22,191.93) .. (43.22,198.71) .. controls (43.22,205.5) and (37.91,211) .. (31.36,211) .. controls (24.81,211) and (19.5,205.5) .. (19.5,198.71) -- cycle ;
	\draw  [fill={rgb, 255:red, 255; green, 255; blue, 255 }  ,fill opacity=1 ] (115.14,31.29) .. controls (115.14,24.5) and (120.45,19) .. (127,19) .. controls (133.55,19) and (138.86,24.5) .. (138.86,31.29) .. controls (138.86,38.07) and (133.55,43.58) .. (127,43.58) .. controls (120.45,43.58) and (115.14,38.07) .. (115.14,31.29) -- cycle ;

	\draw (349.2,36.36) node  [align=left] {$\displaystyle 1$};
	\draw (502.59,37.19) node  [align=left] {$\displaystyle 0$};
	\draw (502.59,189.14) node  [align=left] {$\displaystyle 3$};
	\draw (347.61,189.14) node  [align=left] {$\displaystyle 2$};
	\draw (126.63,30.52) node  [align=left] {$\displaystyle 0$};
	\draw (30.99,197.94) node  [align=left] {$\displaystyle 1$};
	\draw (223.01,198.71) node  [align=left] {$\displaystyle 2$};
	\draw (594,144) node   {$\cdots $};
	\end{tikzpicture}
	\caption{A construction of a conjunctive network with firing memory which admits attractors of non polynomial period. The interaction graph is a connected union of complete graphs $K_{p_i +1 }$ with $p_i$ a prime number. Each component is setted to an initial condition according to Proposition \ref{prop:cycles} and the connections between components are arbitrary. Global initial condition $x$ is defined such that nodes in initial state $0$ are not allowed to be connected.}
	\label{fig:nonpoly}
\end{figure}

One natural question that arise in the context of the last theorem is if we can say something about the period of the attractors in the case when we restrict a conjunctive network with firing memory to have the same delay $\tau$ in all its coordinates. Would there exists a network of this class which dynamics allows attractors with non-polynomial period? The answer is yes, but in order to exhibit it, we need to prove a proposition that is analogous to Proposition \ref{prop:cycles} (complete proof in the Appendix \ref{sec:nonpolyproof}). 
\begin{proposition}
	Let $\tau \geq 2$.  For every integer $k\geq2$, there exists a conjunctive network with firing memory and maximum delay $dt_i = \tau$ in every node $i$ which admits attractors with period $k(\tau+1).$
	\label{propcycleshom}
\end{proposition}
A gadget that we use in the proof of the latter result and that is very important for our main result as well is the \emph{block}. To define a block, let us first define $C = K_{\tau +1}$ as complete graph  with $\tau + 1$ vertices. We recall that these gadget defines a conjunctive network with firing memory which allows cycles of length $\tau +1$ when we have that the maximum delay of every node is $dt_i = \tau$ (see Figure \ref{fig:example}). We are going to call this structure a \emph{clock}.  We define a block $B$ as a $\tau+1$-path such that every node has a $\tau -1$ neighbours in a different clock beside its neighbour in the path as it is shown in Figure \ref{fig:structblock}.  \\[7pt]
An immediate consequence of Proposition \ref{propcycleshom} is that we can exhibit a conjunctive network with firing memory (that have the same maximum delay values in every coordinate) which admits attractors with non polynomial period. This is possible because we can replicate, in this context, the same strategy we used to prove Theorem \ref{teononpolydif}. The complete proof of this result is available in the Appendix \ref{sec:nonpolyproof}.
	\begin{figure}[t]
	\centering		
	\tikzset{every picture/.style={line width=0.75pt}} 
	
	\begin{tikzpicture}[x=0.75pt,y=0.75pt,yscale=-0.5,xscale=0.5]
	
	\draw    (481.5,171) -- (541,171) ;

	\draw    (393,171) -- (452.5,171) ;

	\draw    (172,171) -- (247.5,171) ;

	\draw    (62,171) -- (137.5,171) ;

	\draw  [fill={rgb, 255:red, 255; green, 255; blue, 255 }  ,fill opacity=1 ] (27,151) -- (97,151) -- (97,191) -- (27,191) -- cycle ;
	\draw  [fill={rgb, 255:red, 255; green, 255; blue, 255 }  ,fill opacity=1 ] (137,151) -- (207,151) -- (207,191) -- (137,191) -- cycle ;
	\draw    (283,171) -- (358.5,171) ;

	\draw  [fill={rgb, 255:red, 255; green, 255; blue, 255 }  ,fill opacity=1 ] (248,151) -- (318,151) -- (318,191) -- (248,191) -- cycle ;
	\draw  [fill={rgb, 255:red, 255; green, 255; blue, 255 }  ,fill opacity=1 ] (358,151) -- (428,151) -- (428,191) -- (358,191) -- cycle ;
	\draw    (541,171) -- (616.5,171) ;

	\draw  [fill={rgb, 255:red, 255; green, 255; blue, 255 }  ,fill opacity=1 ] (506,151) -- (576,151) -- (576,191) -- (506,191) -- cycle ;
	\draw   (109.5,76) .. controls (109.55,71.33) and (107.24,68.98) .. (102.57,68.93) -- (70.57,68.61) .. controls (63.9,68.54) and (60.59,66.17) .. (60.64,61.5) .. controls (60.59,66.17) and (57.24,68.47) .. (50.57,68.4)(53.57,68.43) -- (18.57,68.08) .. controls (13.9,68.03) and (11.55,70.34) .. (11.5,75.01) ;
	\draw    (27.5,92) -- (27,151) ;

	\draw  [fill={rgb, 255:red, 255; green, 255; blue, 255 }  ,fill opacity=1 ] (12,92) .. controls (12,83.44) and (18.94,76.5) .. (27.5,76.5) .. controls (36.06,76.5) and (43,83.44) .. (43,92) .. controls (43,100.56) and (36.06,107.5) .. (27.5,107.5) .. controls (18.94,107.5) and (12,100.56) .. (12,92) -- cycle ;
	\draw    (96.5,93) -- (96,152) ;

	\draw  [fill={rgb, 255:red, 255; green, 255; blue, 255 }  ,fill opacity=1 ] (81,93) .. controls (81,84.44) and (87.94,77.5) .. (96.5,77.5) .. controls (105.06,77.5) and (112,84.44) .. (112,93) .. controls (112,101.56) and (105.06,108.5) .. (96.5,108.5) .. controls (87.94,108.5) and (81,101.56) .. (81,93) -- cycle ;
	
	\draw    (137.5,91) -- (137,150) ;

	\draw    (206.5,92) -- (206,151) ;

	\draw    (248.5,91) -- (248,150) ;

	\draw    (317.5,92) -- (317,151) ;

	\draw    (359.5,91) -- (359,150) ;

	\draw    (428.5,92) -- (428,151) ;

	\draw    (506.5,91) -- (506,150) ;

	\draw    (575.5,92) -- (575,151) ;

	\draw  [fill={rgb, 255:red, 255; green, 255; blue, 255 }  ,fill opacity=1 ] (122,91) .. controls (122,82.44) and (128.94,75.5) .. (137.5,75.5) .. controls (146.06,75.5) and (153,82.44) .. (153,91) .. controls (153,99.56) and (146.06,106.5) .. (137.5,106.5) .. controls (128.94,106.5) and (122,99.56) .. (122,91) -- cycle ;
	\draw  [fill={rgb, 255:red, 255; green, 255; blue, 255 }  ,fill opacity=1 ] (191,92) .. controls (191,83.44) and (197.94,76.5) .. (206.5,76.5) .. controls (215.06,76.5) and (222,83.44) .. (222,92) .. controls (222,100.56) and (215.06,107.5) .. (206.5,107.5) .. controls (197.94,107.5) and (191,100.56) .. (191,92) -- cycle ;
	\draw  [fill={rgb, 255:red, 255; green, 255; blue, 255 }  ,fill opacity=1 ] (233,91) .. controls (233,82.44) and (239.94,75.5) .. (248.5,75.5) .. controls (257.06,75.5) and (264,82.44) .. (264,91) .. controls (264,99.56) and (257.06,106.5) .. (248.5,106.5) .. controls (239.94,106.5) and (233,99.56) .. (233,91) -- cycle ;
	\draw  [fill={rgb, 255:red, 255; green, 255; blue, 255 }  ,fill opacity=1 ] (302,92) .. controls (302,83.44) and (308.94,76.5) .. (317.5,76.5) .. controls (326.06,76.5) and (333,83.44) .. (333,92) .. controls (333,100.56) and (326.06,107.5) .. (317.5,107.5) .. controls (308.94,107.5) and (302,100.56) .. (302,92) -- cycle ;
	\draw  [fill={rgb, 255:red, 255; green, 255; blue, 255 }  ,fill opacity=1 ] (344,91) .. controls (344,82.44) and (350.94,75.5) .. (359.5,75.5) .. controls (368.06,75.5) and (375,82.44) .. (375,91) .. controls (375,99.56) and (368.06,106.5) .. (359.5,106.5) .. controls (350.94,106.5) and (344,99.56) .. (344,91) -- cycle ;
	\draw  [fill={rgb, 255:red, 255; green, 255; blue, 255 }  ,fill opacity=1 ] (413,92) .. controls (413,83.44) and (419.94,76.5) .. (428.5,76.5) .. controls (437.06,76.5) and (444,83.44) .. (444,92) .. controls (444,100.56) and (437.06,107.5) .. (428.5,107.5) .. controls (419.94,107.5) and (413,100.56) .. (413,92) -- cycle ;
	\draw  [fill={rgb, 255:red, 255; green, 255; blue, 255 }  ,fill opacity=1 ] (491,91) .. controls (491,82.44) and (497.94,75.5) .. (506.5,75.5) .. controls (515.06,75.5) and (522,82.44) .. (522,91) .. controls (522,99.56) and (515.06,106.5) .. (506.5,106.5) .. controls (497.94,106.5) and (491,99.56) .. (491,91) -- cycle ;
	\draw  [fill={rgb, 255:red, 255; green, 255; blue, 255 }  ,fill opacity=1 ] (560,92) .. controls (560,83.44) and (566.94,76.5) .. (575.5,76.5) .. controls (584.06,76.5) and (591,83.44) .. (591,92) .. controls (591,100.56) and (584.06,107.5) .. (575.5,107.5) .. controls (566.94,107.5) and (560,100.56) .. (560,92) -- cycle ;
	\draw   (20.5,200) .. controls (20.51,204.67) and (22.84,207) .. (27.51,206.99) -- (296.51,206.52) .. controls (303.18,206.51) and (306.51,208.83) .. (306.52,213.5) .. controls (306.51,208.83) and (309.84,206.49) .. (316.51,206.48)(313.51,206.49) -- (585.51,206.01) .. controls (590.18,206) and (592.51,203.67) .. (592.5,199) ;

	\draw (466,169) node   {$\cdots $};
	\draw (62,89) node   {$\cdots $};
	\draw (172,88) node   {$\cdots $};
	\draw (283,88) node   {$\cdots $};
	\draw (394,88) node   {$\cdots $};
	\draw (541,88) node   {$\cdots $};
	\draw (63,44) node   {$\tau \ -1\ \text{clocks}$};
	\draw (322,229) node   {$\tau \ +1\ \text{nodes}$};
	\end{tikzpicture}
	\caption{Structure of a block used in Proposition \ref{prop:cycles} (see Appendix Section \ref{sec:nonpolyproof},  Proposition \ref{Appendix:propcycle} for the complete proof) to define a conjunctive network with firing memory and maximum delay values $dt_i = \tau$ for every node $i$ that admits attractors with period $k(\tau+1)$. Every circle in the figure represents clock  associated to a node $r$ represented by a square. This gadget has $\tau + 1$ nodes and every node has $\tau -1$ clocks.}
	\label{fig:structblock}
\end{figure}
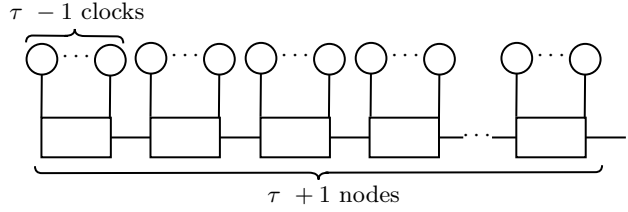	
\begin{theorem}
	Let $\tau \geq 2$.  There exists a conjunctive network with firing memory and maximum delay  $dt_i = \tau$ in every node $i$ which admits attractors with non polynomial period.
	\label{teo:nonpolysamedelay}
\end{theorem}
\newpage
\section{\textsc{2-And-Prediction} is \textbf{PSPACE}-complete.}	
\label{sec:PSPACE}
The fact that there exists conjunctive networks with firing memory that admits attractors with non polynomial period together with the structure of the gadgets we described, strongly suggest that there are conjunctive networks that are able to simulate boolean circuits. The next results confirm this insight establishing that conjunctive networks with firing memory are able to simulate iterated boolean circuits. We remark that previous results on the attractors period hold for an arbitrary value for the maximum delay that is defined for all the nodes in the given networks. So, we address the capability of conjunctive networks with firing  memory and maximum delay $dt_i = 2$ in every node $i$ to simulate an arbitrary iterated boolean circuit. These results have consequences related to computational complexity and universality of conjunctive networks with firing memory.
\vspace{-1.2in}
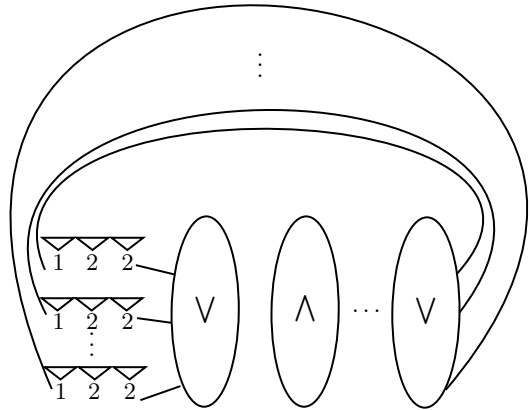
\begin{figure}[htp]
\centering

\tikzset{every picture/.style={line width=0.75pt}} 

\begin{tikzpicture}[x=1pt,y=0.4pt,yscale=-1.4,xscale=0.78]
\path (100,280); 

\draw   (211.11,168.43) -- (203.43,160.14) -- (218.8,160.14) -- cycle ;
\draw   (227.69,168.43) -- (220.01,160.14) -- (235.37,160.14) -- cycle ;
\draw   (244.61,168.43) -- (236.93,160.14) -- (252.3,160.14) -- cycle ;
\draw   (212.21,255.92) -- (204.53,247.63) -- (219.89,247.63) -- cycle ;
\draw   (228.79,255.92) -- (221.1,247.63) -- (236.47,247.63) -- cycle ;
\draw   (245.71,255.92) -- (238.02,247.63) -- (253.39,247.63) -- cycle ;
\draw   (266.05,210.03) .. controls (266.35,174.5) and (273.88,145.7) .. (282.87,145.7) .. controls (291.87,145.7) and (298.93,174.5) .. (298.64,210.03) .. controls (298.35,245.56) and (290.82,274.37) .. (281.82,274.37) .. controls (272.82,274.37) and (265.76,245.56) .. (266.05,210.03) -- cycle ;
\draw   (314.51,210.03) .. controls (314.8,174.5) and (322.33,145.7) .. (331.33,145.7) .. controls (340.32,145.7) and (347.38,174.5) .. (347.09,210.03) .. controls (346.8,245.56) and (339.27,274.37) .. (330.27,274.37) .. controls (321.27,274.37) and (314.22,245.56) .. (314.51,210.03) -- cycle ;
\draw   (373.09,210.67) .. controls (373.38,175.14) and (380.91,146.34) .. (389.91,146.34) .. controls (398.91,146.34) and (405.96,175.14) .. (405.67,210.67) .. controls (405.38,246.2) and (397.85,275) .. (388.85,275) .. controls (379.85,275) and (372.8,246.2) .. (373.09,210.67) -- cycle ;
\draw    (204.66,181.83) .. controls (161.96,46.19) and (483.89,62.67) .. (404.5,184.36) ;

\draw    (205.26,212.11) .. controls (135.11,23.55) and (504.14,32.25) .. (405.67,210.67) ;

\draw    (207.98,262.43) .. controls (84.72,-133.39) and (577.71,-29.9) .. (399.02,262.96) ;

\draw    (248.74,178.66) -- (267.62,185) ;

\draw    (249.28,214.46) -- (265.98,217.63) ;

\draw    (250.93,269.93) -- (270.36,260.42) ;

\draw   (210.73,209.31) -- (203.04,201.02) -- (218.41,201.02) -- cycle ;
\draw   (227.3,209.31) -- (219.62,201.02) -- (234.99,201.02) -- cycle ;
\draw   (244.22,209.31) -- (236.54,201.02) -- (251.91,201.02) -- cycle ;

\draw (211.29,177.08) node  [align=left] {$\displaystyle 1$};
\draw (244.78,177.08) node  [align=left] {$\displaystyle 2$};
\draw (227.52,177.08) node  [align=left] {$\displaystyle 2$};
\draw (211.29,216.81) node  [align=left] {$\displaystyle 1$};
\draw (244.78,216.81) node  [align=left] {$\displaystyle 2$};
\draw (227.52,216.81) node  [align=left] {$\displaystyle 2$};
\draw (227.45,232.84-6) node   {$\vdots $};
\draw (212.38,263.78) node  [align=left] {$\displaystyle 1$};
\draw (245.88,263.78) node  [align=left] {$\displaystyle 2$};
\draw (228.61,263.78) node  [align=left] {$\displaystyle 2$};
\draw (361.52,209.72) node   {$\dotsc $};
\draw (330.8,206.89) node   {$\bigwedge $};
\draw (282.35,210.03) node [rotate=-180]  {$\bigwedge $};
\draw (309.05,37.87) node   {$\vdots $};
\draw (389.38,210.67) node [rotate=-180]  {$\bigwedge $};

\end{tikzpicture}

	\caption{Interaction graph associated to a conjunctive network with firing memory and maximum delay vector $dt_i = 2$ for every node $i$ that simulates an iterated monotone boolean circuit $C$.  Layers are made up by AND or OR gates exclusively, using the gadgets shown in Figure \ref{fig:AND} and Figure \ref{fig:OR} are alternately ordered.}
	\label{fig:circuit}
\end{figure}
\vspace{-0.3in}
\begin{proposition}
	For every monotone boolean circuit $C: \{0,1\}^n \to \{0,1\}^n$ there is a conjunctive network $F$ with firing memory such that:  i) its interaction graph $G$ has polynomial size in $n$, ii) its maximum delay values are  $dt_i = 2$ in every node $i\in V(G)$ and iii) $F$ simulates every iteration of  $C$ in linear time.
	\label{prop:circand}
\end{proposition}
\begin{proof}
	Let $C$ be a monotone circuit. We assume that: i) every input has out degree $1$, ii) every output is identified with an input, iii) the degree of every logic gate in C is 4, and iv) every layer contains exclusively OR or AND gates and they are ordered alternately, i.e., if a the $k$-th layer is made up of AND gates then the $k+1$-th is made up of OR gates (See Proposition \ref{prop:PSPACECirc}). We represent this structure using the block gadget we defined for the last propositions (see Figure \ref{fig:structblock}) with maximum delay vector $dt_i = 2$ in every node $i$. A scheme of the interaction graph associated to this conjunctive network with firing memory is shown in Figure \ref{fig:circuit}. Let $i_1, \hdots, i_n$ be the inputs of $C$ and because we are considering $C$ as an iterated circuit, we are going to identify its outputs by the same names. For every $k \in \{1,\hdots,n\}$, we define a block $B_{i_k}$.  These blocks are made up of a path of length three in which every node is connected to a clock. We introduce the following notation: if $x \in \{0,1\}^n$ and $B$ is a block then $x_{B} = x_ux_vx_w$ where $u,v$ and $w$ are the labels of the vertices in the path. Given a initial configuration $x \in \{0,1\}^n$ for the circuit $C$,  we code it using a variable $y \in \{0,1\}^m$,  defined by maping the blocks in the following way: $$ y_{B_{i_k} } = \begin{cases}
	122 \text{ if } x_{i_k} = 0 \\
	120 \text{ if } x_{i_k} = 1 \\
	\end{cases}$$ 
	For every logic gate AND or OR we define a gadget as the one showed in Figure \ref{fig:AND} and Figure \ref{fig:OR}. Note that every logic gate is represented by a block so the last coding is well defined. We remark that as all the gadget have a constant number of nodes and edges then, $m = O(n)$ and because of the assumptions we are doing on $C$ we have that this coding has polynomial size in $n$. We will call $\varphi$ the function such that $\varphi(x) = y$. Finally, as it is shown in Figures \ref{fig:iterAND} and Figure \ref{fig:iterOR} the information is transmitted through the blocks in a way such that in a maximum of $6$ steps the nodes return to the initial condition so the circuit is cleared and the structures are available for continue receiving and emitting signals. Then, we have given $x \in \{0,1\}^n$ and $y = \varphi(x)$, there exists $p\geq 1$ such that $\varphi(C^t (x))_{B_{i_k}} =  (F^{pt}(y))_{B_{i_k}}$ Thus, the conjunctive network defined using these gadgets simulates $C$ in polynomial space and linear time.
\end{proof}
\begin{figure}[h]
\tikzset{every picture/.style={line width=0.75pt}} 
\centering    
\begin{tikzpicture}[x=0.75pt,y=0.75pt,yscale=-0.62,xscale=0.62]

\draw    (292.5,123) -- (325,123) ;

\draw    (338.5,123) -- (371,123) ;

\draw    (284.5,89) -- (284,116) ;

\draw    (332.5,89) -- (332,116) ;

\draw    (381.5,89) -- (381,116) ;

\draw   (284.5,89) -- (262.25,70) -- (306.75,70) -- cycle ;
\draw   (332.5,89) -- (310.25,70) -- (354.75,70) -- cycle ;
\draw   (381.5,89) -- (359.25,70) -- (403.75,70) -- cycle ;
\draw    (194.5,171) -- (234.5,136) ;

\draw    (232.5,83) -- (238.5,112) ;

\draw    (498.5,74) -- (392.5,122) ;

\draw    (393.5,130) -- (493.5,169) ;

\draw    (46.5,72) -- (79,72) ;

\draw    (92.5,72) -- (125,72) ;

\draw    (38.5,36) -- (38,63) ;

\draw    (86.5,36) -- (86,63) ;

\draw    (135.5,36) -- (135,63) ;

\draw   (38.5,36) -- (16.25,17) -- (60.75,17) -- cycle ;
\draw   (86.5,36) -- (64.25,17) -- (108.75,17) -- cycle ;
\draw   (135.5,36) -- (113.25,17) -- (157.75,17) -- cycle ;
\draw    (47.5,170) -- (80,170) ;

\draw    (93.5,170) -- (126,170) ;

\draw    (39.5,134) -- (39,161) ;

\draw    (87.5,134) -- (87,161) ;

\draw    (136.5,134) -- (136,161) ;

\draw   (39.5,134) -- (17.25,115) -- (61.75,115) -- cycle ;
\draw   (87.5,134) -- (65.25,115) -- (109.75,115) -- cycle ;
\draw   (136.5,134) -- (114.25,115) -- (158.75,115) -- cycle ;

\draw    (515.5,72) -- (548,72) ;

\draw    (561.5,72) -- (594,72) ;

\draw    (507.5,36) -- (507,63) ;

\draw    (555.5,36) -- (555,63) ;

\draw    (604.5,36) -- (604,63) ;

\draw   (507.5,36) -- (485.25,17) -- (529.75,17) -- cycle ;
\draw   (555.5,36) -- (533.25,17) -- (577.75,17) -- cycle ;
\draw   (604.5,36) -- (582.25,17) -- (626.75,17) -- cycle ;

\draw    (515.5,171) -- (548,171) ;

\draw    (561.5,171) -- (594,171) ;

\draw    (507.5,135) -- (507,162) ;

\draw    (555.5,135) -- (555,162) ;

\draw    (604.5,135) -- (604,162) ;

\draw   (507.5,135) -- (485.25,116) -- (529.75,116) -- cycle ;
\draw   (555.5,135) -- (533.25,116) -- (577.75,116) -- cycle ;
\draw   (604.5,135) -- (582.25,116) -- (626.75,116) -- cycle ;

\draw    (251.5,123) -- (275.5,123) ;

\draw    (142.5,72) -- (175,72) ;

\draw    (188.5,72) -- (221,72) ;

\draw    (182.5,36) -- (182,63) ;

\draw    (231.5,36) -- (231,63) ;

\draw   (182.5,36) -- (160.25,17) -- (204.75,17) -- cycle ;
\draw   (231.5,36) -- (209.25,17) -- (253.75,17) -- cycle ;
\draw    (143.5,170) -- (176,170) ;

\draw    (186.5,134) -- (186,161) ;

\draw   (186.5,134) -- (164.25,115) -- (208.75,115) -- cycle ;
\draw  [dash pattern={on 4.5pt off 4.5pt}] (15.5,8) -- (159.5,8) -- (159.5,90) -- (15.5,90) -- cycle ;
\draw  [dash pattern={on 4.5pt off 4.5pt}] (15.25,106.5) -- (159.25,106.5) -- (159.25,188.5) -- (15.25,188.5) -- cycle ;

\draw (285,123) node  [align=left] {$\displaystyle 1$};
\draw (382,123) node  [align=left] {$\displaystyle 2$};
\draw (332,123) node  [align=left] {$\displaystyle 2$};
\draw (39,72) node  [align=left] {$\displaystyle 1$};
\draw (136,72) node  [align=left] {$\displaystyle 2$};
\draw (86,72) node  [align=left] {$\displaystyle 2$};
\draw (286,78) node  [align=left] {$\displaystyle 2$};
\draw (333,77) node  [align=left] {$\displaystyle 0$};
\draw (382,77) node  [align=left] {$\displaystyle 1$};
\draw (40,25) node  [align=left] {$\displaystyle 2$};
\draw (87,24) node  [align=left] {$\displaystyle 0$};
\draw (136,24) node  [align=left] {$\displaystyle 1$};
\draw (40,170) node  [align=left] {$\displaystyle 1$};
\draw (137,170) node  [align=left] {$\displaystyle 2$};
\draw (87,170) node  [align=left] {$\displaystyle 2$};
\draw (41,123) node  [align=left] {$\displaystyle 2$};
\draw (88,122) node  [align=left] {$\displaystyle 0$};
\draw (137,122) node  [align=left] {$\displaystyle 1$};
\draw (508,72) node  [align=left] {$\displaystyle 1$};
\draw (605,72) node  [align=left] {$\displaystyle 2$};
\draw (555,72) node  [align=left] {$\displaystyle 2$};
\draw (509,25) node  [align=left] {$\displaystyle 2$};
\draw (556,24) node  [align=left] {$\displaystyle 0$};
\draw (605,24) node  [align=left] {$\displaystyle 1$};
\draw (508,171) node  [align=left] {$\displaystyle 1$};
\draw (605,171) node  [align=left] {$\displaystyle 2$};
\draw (555,171) node  [align=left] {$\displaystyle 2$};
\draw (509,124) node  [align=left] {$\displaystyle 2$};
\draw (556,123) node  [align=left] {$\displaystyle 0$};
\draw (605,123) node  [align=left] {$\displaystyle 1$};
\draw (241,123) node  [align=left] {$\displaystyle 2$};
\draw (232,72) node  [align=left] {$\displaystyle 2$};
\draw (182,72) node  [align=left] {$\displaystyle 1$};
\draw (183,24) node  [align=left] {$\displaystyle 2$};
\draw (232,24) node  [align=left] {$\displaystyle 0$};
\draw (187,170) node  [align=left] {$\displaystyle 1$};
\draw (187,122) node  [align=left] {$\displaystyle 2$};
\end{tikzpicture}
	\caption{Gadget of AND gates used in the graph shown in Figure \ref{fig:circuit}. Signals are transmitted and coded based on the block gadget.}
		\label{fig:AND}
\end{figure}
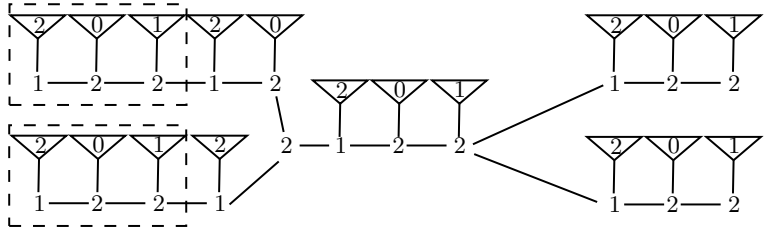
\begin{figure}[H]
\tikzset{every picture/.style={line width=0.75pt}} 
\centering
\begin{tikzpicture}[x=0.75pt,y=0.75pt,yscale=-0.62,xscale=0.62]

\draw    (287.5,126) -- (320,126) ;

\draw    (333.5,126) -- (366,126) ;

\draw    (279.5,92) -- (279,119) ;

\draw    (327.5,92) -- (327,119) ;

\draw    (376.5,92) -- (376,119) ;

\draw   (279.5,92) -- (257.25,73) -- (301.75,73) -- cycle ;
\draw   (327.5,92) -- (305.25,73) -- (349.75,73) -- cycle ;
\draw   (376.5,92) -- (354.25,73) -- (398.75,73) -- cycle ;
\draw    (493.5,77) -- (387.5,125) ;

\draw    (388.5,133) -- (488.5,172) ;

\draw    (71.5,83) -- (104,83) ;

\draw    (117.5,83) -- (150,83) ;

\draw    (63.5,47) -- (63,74) ;

\draw    (111.5,47) -- (111,74) ;

\draw    (160.5,47) -- (160,74) ;

\draw   (63.5,47) -- (41.25,28) -- (85.75,28) -- cycle ;
\draw   (111.5,47) -- (89.25,28) -- (133.75,28) -- cycle ;
\draw   (160.5,47) -- (138.25,28) -- (182.75,28) -- cycle ;

\draw    (68.5,178) -- (101,178) ;

\draw    (114.5,178) -- (147,178) ;

\draw    (60.5,142) -- (60,169) ;

\draw    (108.5,142) -- (108,169) ;

\draw    (157.5,142) -- (157,169) ;

\draw   (60.5,142) -- (38.25,123) -- (82.75,123) -- cycle ;
\draw   (108.5,142) -- (86.25,123) -- (130.75,123) -- cycle ;
\draw   (157.5,142) -- (135.25,123) -- (179.75,123) -- cycle ;

\draw    (510.5,75) -- (543,75) ;

\draw    (556.5,75) -- (589,75) ;

\draw    (502.5,39) -- (502,66) ;

\draw    (550.5,39) -- (550,66) ;

\draw    (599.5,39) -- (599,66) ;

\draw   (502.5,39) -- (480.25,20) -- (524.75,20) -- cycle ;
\draw   (550.5,39) -- (528.25,20) -- (572.75,20) -- cycle ;
\draw   (599.5,39) -- (577.25,20) -- (621.75,20) -- cycle ;

\draw    (510.5,174) -- (543,174) ;

\draw    (556.5,174) -- (589,174) ;

\draw    (502.5,138) -- (502,165) ;

\draw    (550.5,138) -- (550,165) ;

\draw    (599.5,138) -- (599,165) ;

\draw   (502.5,138) -- (480.25,119) -- (524.75,119) -- cycle ;
\draw   (550.5,138) -- (528.25,119) -- (572.75,119) -- cycle ;
\draw   (599.5,138) -- (577.25,119) -- (621.75,119) -- cycle ;

\draw    (165.5,179) -- (271.5,131) ;
\draw    (270.5,123) -- (170.5,84) ;
\draw (280,126) node  [align=left] {$\displaystyle 1$};
\draw (377,126) node  [align=left] {$\displaystyle 2$};
\draw (327,126) node  [align=left] {$\displaystyle 2$};
\draw (64,83) node  [align=left] {$\displaystyle 1$};
\draw (161,83) node  [align=left] {$\displaystyle 2$};
\draw (111,83) node  [align=left] {$\displaystyle 2$};
\draw (281,81) node  [align=left] {$\displaystyle 2$};
\draw (328,80) node  [align=left] {$\displaystyle 0$};
\draw (377,80) node  [align=left] {$\displaystyle 1$};
\draw (65,36) node  [align=left] {$\displaystyle 2$};
\draw (112,35) node  [align=left] {$\displaystyle 0$};
\draw (161,35) node  [align=left] {$\displaystyle 1$};
\draw (503,174) node  [align=left] {$\displaystyle 1$};
\draw (600,174) node  [align=left] {$\displaystyle 2$};
\draw (550,174) node  [align=left] {$\displaystyle 2$};
\draw (504,127) node  [align=left] {$\displaystyle 2$};
\draw (551,126) node  [align=left] {$\displaystyle 0$};
\draw (600,126) node  [align=left] {$\displaystyle 1$};
\draw (503,75) node  [align=left] {$\displaystyle 1$};
\draw (600,75) node  [align=left] {$\displaystyle 2$};
\draw (550,75) node  [align=left] {$\displaystyle 2$};
\draw (504,28) node  [align=left] {$\displaystyle 2$};
\draw (551,27) node  [align=left] {$\displaystyle 0$};
\draw (600,27) node  [align=left] {$\displaystyle 1$};
\draw (61,178) node  [align=left] {$\displaystyle 1$};
\draw (158,178) node  [align=left] {$\displaystyle 2$};
\draw (108,178) node  [align=left] {$\displaystyle 2$};
\draw (62,131) node  [align=left] {$\displaystyle 2$};
\draw (109,130) node  [align=left] {$\displaystyle 0$};
\draw (158,130) node  [align=left] {$\displaystyle 1$};
\end{tikzpicture}	
	\caption{Gadget of OR gates used in the graph shown in Figure \ref{fig:circuit}. Signals are transmitted and coded based on the block gadget.}
	\label{fig:OR}
\end{figure}
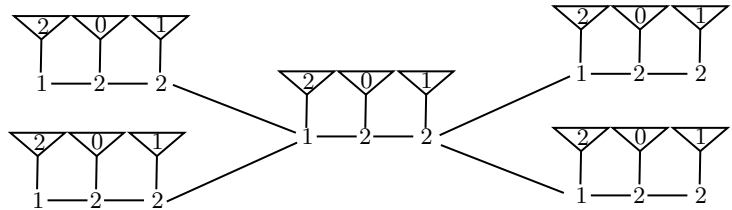
As a direct consequence of the Proposition \ref{prop:circand}, we have that conjunctive networks with firing memory are universal, i.e, they can simulate every boolean network.
Finally, we address the question about the computational complexity of the prediction problem \textsc{2-And-Prediction}.  As a direct corollary of the latter proposition, we have that \textsc{Iter-Mon-Circuit} can be reduced to \textsc{2-And-Prediction} and thus the problem is \textbf{PSPACE}-complete.
\begin{theorem}
	The problem \textsc{2-And-Prediction} is \textbf{PSPACE}-complete.
	\label{teo:PSPACE}
\end{theorem}
\begin{proof}
	It is direct from Proposition \ref{prop:circand} and Proposition \ref{prop:PSPACECirc}
\end{proof}
\section{Discussion}
In this paper, we have studied the effect of an specific type of delay called firing memory in the dynamics of conjunctive boolean networks. More specifically, we have addressed the prediction problem in conjunctive networks with firing memory whom maximum delay is $2$ in every node. We concluded that not only these type of networks admit attractors of non polynomial period but the latter problem turned out to be \textbf{PSPACE}-complete.  Deducing this result was possible because of: i) the capability of conjunctive networks with firing memory, whom have the same value for maximum delay in every node, to have attractors with period proportional to the maximum delay value and ii) the capability of transmitting information through a wire that clears itself once the information has been transmitted. These two main observations about conjunctive networks with firing memory allowed us to deduce the structure of the main gadgets used for the proof of our main results. These properties are quite surprising considering that conjunctive boolean networks admit only attractors of bounded period and the prediction problem is the \textbf{P} class. Moreover, previous result on the effect of firing memory in the  dynamics of the dual version of these type of networks, the disjunctive networks, suggested that firing memory tend to freeze the dynamics of these networks, reducing the period of the possible attractors that the network admits. We remark the relevance of the achieved results as they show that firing memory have effects on the dynamical properties of the original network that are similar to the ones exhibited by other update schemes that are somehow between synchronous and asynchronous dynamics in other type of boolean networks, such as the effects of block sequential update scheme in majority rules. In fact, in the latter case, the prediction problem with parallel update is \textbf{P}-complete while it is \textbf{PSPACE}-complete when we consider a block-sequential update scheme. It might be possible to deduce from the latter observation that firing memory allows to add asynchronicity to the dynamics of a boolean network in a less arbitrary way compared to block sequential update, that needs a predefined partition and an specific partial order.\\[12pt]
An interesting topic for future work is the characterization of the dynamics of conjunctive networks with firing memory. While we have described conjunctive networks with firing memory that admit attractors with period proportional to the maximum delay values, the possibility of the existence of networks admitting attractors with different period (not necessarily proportional to maximum delay values) remains still open for studying. In addition, the effect of firing memory in networks defined by particular topologies such as planar graphs or two dimensional grids might be interesting to analyse. Besides, considering the fact that in the light of the results of this paper there is no clear insight about a general effect of firing memory in a simple class of boolean functions such as AND or OR (we have complex dynamics in one case  and we have dynamics that admit only fixed points in the other), an interesting topic for future work could be studying prediction problems in other classes of boolean networks with firing memory that are somehow similar to conjunctive networks (another functions that are linear for example) such as XOR networks.
\section{Acknowledgement}
 This work has been partially supported by: CONICYT via  PAI + Convocatoria Nacional Subvenci\'on a la Incorporaci\'on en la Academia A\~no 2017 + PAI77170068 (P.M.), PFCHA/DOCTORADO NACIONAL/2018 – 21180910 (M.R.W) and ECOS C16E01 (E.G and M.R.W).
We would also like to thank Alejandro Maass who provided insight and expertise that greatly assisted us in the course of this research.
%
%
\bibliographystyle{splncs04}
\bibliography{bibliography}
\appendix
\newpage
\section{Appendix}
\subsection{Proof of Theorem \ref{teononpolydif} and Theorem \ref{teo:nonpolysamedelay}}
\label{sec:nonpolyproof}
\begin{teo}
	There exists a connected conjunctive network with firing memory (and not necessarily the same values for maximum delay) which admits attractors with non polynomial period. 
\end{teo}
\begin{proof}
	Let us consider a fixed natural number $m \geq 2$ and a collection of prime numbers $p_i, \hdots, p_l, $ such that $2 \leq p_i \leq m$ where $l = \pi (m)$ and $\pi (m)$ denotes the number of primes not exceeding $m$. For each $i \in \{1, \hdots, l\}$ we consider a conjunctive network with interaction graph given by a complete graph $K_{p_i + 1}$ as we did in the proof of the last proposition. We consider a graph $G$ defined as the connected union of the previous complete graphs in the following way: we consider $V = \bigcup \limits_{i=1}^{l} V(K_{p_i +1})$ and for every $i \in \{1,\hdots,l\}$ we choose an arbitrary vertex $s_i \in K_{p_i+1}$ and we consider $E = \bigcup \limits_{i=1}^{l} E(K_{p_i +1})  \cup \bigcup \limits_{i=1}^{l-1} \{s_is_i+1\}$. In other words, we consider the union of the previous complete graphs and we connect each other by a unique aribitrarily chosen edge.  \\
	
	Let us define the label function $\varphi: V \to \{1,\hdots, l\}$ given by $\varphi(u) = i \text{ if } u \in V(K_{p_i +1})$.	Let $F^{p_1}, \hdots, F^{p_l}$ be the conjunctive networks with firing memory associated to each complete graph $K_{p_i +1}$. We recall that each of these functions has a maximum delay value of $p_i$ in every node. We define the following conjunctive network with firing memory $F : \{0,1\}^{|V|} \to \{0,1\}^{|V|}$ given by $F(x)_i = F^{p_{\varphi(i)}}(x)_i$ for every $i \in V \setminus \{s_j| j \in \{1, \hdots, l\} \}$ and $F(x)_{s_i} = F^{p_{\varphi(s_i)}}(x)_{s_i} \wedge x_{s_{i+1}}.$ It is not difficult to see that the interaction graph of $F$ is $G$. \\
	
	Finally, let us define the initial condition $x \in \{0,1\}^{|V|}$ in the following way: for the vertices in $V(K_{p_i +1})$ we assign the initial condition $0123,\hdots,p_i$ with the only restriction that it if $x_{s_i} = 0$ then $x_{s_{i+1}} \neq 0$. Note that we need this because if two nodes are connected and both have initial state $0$ then the global dynamics converge to the fixed point $0$. It is not difficult to see that starting from $x$ every node in $K_{p_i +1}$  is in a cycle with period $p_i$. Thus, we have that if $T$ is the global period of the network then: 	
	\begin{equation*}
	T \geq \prod \limits_{i=1}^{\pi(m)} p_i.
	\end{equation*}
	And also, we have that:	
	\begin{equation}
	|V| = \sum \limits_{i=1}^{\pi(m)} (p_i + 1).
	\label{vertexbound}
	\end{equation}	
	Additionally, if we define $\theta = \sum_{i=1}^{\pi(m)} \log (p_i)$, we have that: 	
	\begin{equation}
	T \geq \exp(\theta(m)).
	\label{Teq}
	\end{equation}
	Based in \ref{vertexbound}, and \ref{Teq}, we are going to apply a technique used in \cite{kiwi1994no} to deduce that T is not polynomial. This is based a result stated in \cite{hardy1979introduction}:
	\begin{equation}
	\begin{split}
	\pi (m) & = \Theta \left(\frac{m}{\log(m)}\right), \\
	\end{split}
	\label{pieq}
	\end{equation}
	\begin{align}
	\theta(m) = \Theta \left(\pi \log(m) \right). 
	\label{thetaeq}
	\end{align}
	
	We observe that from  \ref{pieq} it can be deduced that $m = O(\pi(m) log(m))$ and then $|V|  = O(\pi(m)^2 \log(m))$. On the other hand, using this last observation we deduce that $\log |V| = O(\log(m))$. And finally, $\sqrt{|V|\log(|V|)} = O(\pi(m) \log(m))$ which is equivalent to say that $\pi(m) \log(m) = \Omega(\sqrt{|V|\log(|V|)})$ and thus
	\begin{equation*}
	T \geq \exp(\Omega(\sqrt{|V|\log(|V|)}).
	\end{equation*}
	We conclude that $F$ has attractors with non polynomial period.
\end{proof}

\begin{prop}
	Let $\tau \geq 2$.  For every integer $k\geq2$, there exists a conjunctive network with firing memory and maximum delay $dt_i = \tau$ in every node $i$ which admits attractors with period $k(\tau+1).$
	\label{Appendix:propcycle}
\end{prop}
\begin{proof}
	Let $k \geq 1$ be an integer. Let us define $C = K_{\tau +1}$ as complete graph  with $\tau + 1$ vertices. We recall that these gadget defines a conjunctive network with firing memory which allows cycles of length $\tau +1$ when we have that the maximum delay of every node is $dt_i = \tau$. We are going to call this structure a clock. We are going to define $k$ copies of a certain gadget that we will call block. Let $j \in \{0, \hdots, k\}$ we define the $j$-block $B_j$ as a $\tau+1$-path such that every node has a $\tau -1$ neighbours in a different clock beside its neighbour in the path as it is shown in Figure \ref{fig:appendixstructblock}. As every block contains an induced path then we can number the vertices in the path defining an initial and a terminal vertex.  We write $C(B_j)_{r,l} \text{ for } r= 1,2,3, \hdots, \tau \text{  } l = 1,2,3, \hdots, \tau - 1$ to denote the clocks of the $j$-block that are associated to the $r$-th vertex of the path. Besides, a  node in a clock that is connected to a node in the path is called $C(B_j)_{r,l}a$ or simply $a$ when the context is clear. 
	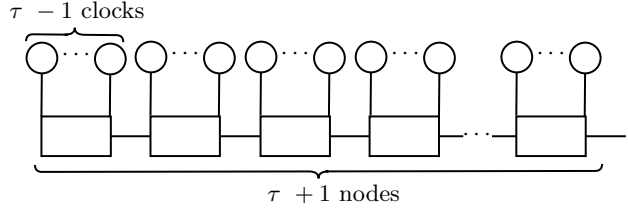
\begin{figure}
		\centering		
		\tikzset{every picture/.style={line width=0.75pt}} 
		
		\begin{tikzpicture}[x=0.75pt,y=0.75pt,yscale=-0.5,xscale=0.5]
		
		\draw    (481.5,171) -- (541,171) ;

		\draw    (393,171) -- (452.5,171) ;

		\draw    (172,171) -- (247.5,171) ;

		\draw    (62,171) -- (137.5,171) ;

		\draw  [fill={rgb, 255:red, 255; green, 255; blue, 255 }  ,fill opacity=1 ] (27,151) -- (97,151) -- (97,191) -- (27,191) -- cycle ;
		\draw  [fill={rgb, 255:red, 255; green, 255; blue, 255 }  ,fill opacity=1 ] (137,151) -- (207,151) -- (207,191) -- (137,191) -- cycle ;
		\draw    (283,171) -- (358.5,171) ;

		\draw  [fill={rgb, 255:red, 255; green, 255; blue, 255 }  ,fill opacity=1 ] (248,151) -- (318,151) -- (318,191) -- (248,191) -- cycle ;
		\draw  [fill={rgb, 255:red, 255; green, 255; blue, 255 }  ,fill opacity=1 ] (358,151) -- (428,151) -- (428,191) -- (358,191) -- cycle ;
		\draw    (541,171) -- (616.5,171) ;

		\draw  [fill={rgb, 255:red, 255; green, 255; blue, 255 }  ,fill opacity=1 ] (506,151) -- (576,151) -- (576,191) -- (506,191) -- cycle ;
		\draw   (109.5,76) .. controls (109.55,71.33) and (107.24,68.98) .. (102.57,68.93) -- (70.57,68.61) .. controls (63.9,68.54) and (60.59,66.17) .. (60.64,61.5) .. controls (60.59,66.17) and (57.24,68.47) .. (50.57,68.4)(53.57,68.43) -- (18.57,68.08) .. controls (13.9,68.03) and (11.55,70.34) .. (11.5,75.01) ;
		\draw    (27.5,92) -- (27,151) ;

		\draw  [fill={rgb, 255:red, 255; green, 255; blue, 255 }  ,fill opacity=1 ] (12,92) .. controls (12,83.44) and (18.94,76.5) .. (27.5,76.5) .. controls (36.06,76.5) and (43,83.44) .. (43,92) .. controls (43,100.56) and (36.06,107.5) .. (27.5,107.5) .. controls (18.94,107.5) and (12,100.56) .. (12,92) -- cycle ;
		\draw    (96.5,93) -- (96,152) ;

		\draw  [fill={rgb, 255:red, 255; green, 255; blue, 255 }  ,fill opacity=1 ] (81,93) .. controls (81,84.44) and (87.94,77.5) .. (96.5,77.5) .. controls (105.06,77.5) and (112,84.44) .. (112,93) .. controls (112,101.56) and (105.06,108.5) .. (96.5,108.5) .. controls (87.94,108.5) and (81,101.56) .. (81,93) -- cycle ;
		
		\draw    (137.5,91) -- (137,150) ;

		\draw    (206.5,92) -- (206,151) ;

		\draw    (248.5,91) -- (248,150) ;

		\draw    (317.5,92) -- (317,151) ;

		\draw    (359.5,91) -- (359,150) ;

		\draw    (428.5,92) -- (428,151) ;

		\draw    (506.5,91) -- (506,150) ;

		\draw    (575.5,92) -- (575,151) ;

		\draw  [fill={rgb, 255:red, 255; green, 255; blue, 255 }  ,fill opacity=1 ] (122,91) .. controls (122,82.44) and (128.94,75.5) .. (137.5,75.5) .. controls (146.06,75.5) and (153,82.44) .. (153,91) .. controls (153,99.56) and (146.06,106.5) .. (137.5,106.5) .. controls (128.94,106.5) and (122,99.56) .. (122,91) -- cycle ;
		\draw  [fill={rgb, 255:red, 255; green, 255; blue, 255 }  ,fill opacity=1 ] (191,92) .. controls (191,83.44) and (197.94,76.5) .. (206.5,76.5) .. controls (215.06,76.5) and (222,83.44) .. (222,92) .. controls (222,100.56) and (215.06,107.5) .. (206.5,107.5) .. controls (197.94,107.5) and (191,100.56) .. (191,92) -- cycle ;
		\draw  [fill={rgb, 255:red, 255; green, 255; blue, 255 }  ,fill opacity=1 ] (233,91) .. controls (233,82.44) and (239.94,75.5) .. (248.5,75.5) .. controls (257.06,75.5) and (264,82.44) .. (264,91) .. controls (264,99.56) and (257.06,106.5) .. (248.5,106.5) .. controls (239.94,106.5) and (233,99.56) .. (233,91) -- cycle ;
		\draw  [fill={rgb, 255:red, 255; green, 255; blue, 255 }  ,fill opacity=1 ] (302,92) .. controls (302,83.44) and (308.94,76.5) .. (317.5,76.5) .. controls (326.06,76.5) and (333,83.44) .. (333,92) .. controls (333,100.56) and (326.06,107.5) .. (317.5,107.5) .. controls (308.94,107.5) and (302,100.56) .. (302,92) -- cycle ;
		\draw  [fill={rgb, 255:red, 255; green, 255; blue, 255 }  ,fill opacity=1 ] (344,91) .. controls (344,82.44) and (350.94,75.5) .. (359.5,75.5) .. controls (368.06,75.5) and (375,82.44) .. (375,91) .. controls (375,99.56) and (368.06,106.5) .. (359.5,106.5) .. controls (350.94,106.5) and (344,99.56) .. (344,91) -- cycle ;
		\draw  [fill={rgb, 255:red, 255; green, 255; blue, 255 }  ,fill opacity=1 ] (413,92) .. controls (413,83.44) and (419.94,76.5) .. (428.5,76.5) .. controls (437.06,76.5) and (444,83.44) .. (444,92) .. controls (444,100.56) and (437.06,107.5) .. (428.5,107.5) .. controls (419.94,107.5) and (413,100.56) .. (413,92) -- cycle ;
		\draw  [fill={rgb, 255:red, 255; green, 255; blue, 255 }  ,fill opacity=1 ] (491,91) .. controls (491,82.44) and (497.94,75.5) .. (506.5,75.5) .. controls (515.06,75.5) and (522,82.44) .. (522,91) .. controls (522,99.56) and (515.06,106.5) .. (506.5,106.5) .. controls (497.94,106.5) and (491,99.56) .. (491,91) -- cycle ;
		\draw  [fill={rgb, 255:red, 255; green, 255; blue, 255 }  ,fill opacity=1 ] (560,92) .. controls (560,83.44) and (566.94,76.5) .. (575.5,76.5) .. controls (584.06,76.5) and (591,83.44) .. (591,92) .. controls (591,100.56) and (584.06,107.5) .. (575.5,107.5) .. controls (566.94,107.5) and (560,100.56) .. (560,92) -- cycle ;
		\draw   (20.5,200) .. controls (20.51,204.67) and (22.84,207) .. (27.51,206.99) -- (296.51,206.52) .. controls (303.18,206.51) and (306.51,208.83) .. (306.52,213.5) .. controls (306.51,208.83) and (309.84,206.49) .. (316.51,206.48)(313.51,206.49) -- (585.51,206.01) .. controls (590.18,206) and (592.51,203.67) .. (592.5,199) ;

		\draw (466,169) node   {$\cdots $};
		\draw (62,89) node   {$\cdots $};
		\draw (172,88) node   {$\cdots $};
		\draw (283,88) node   {$\cdots $};
		\draw (394,88) node   {$\cdots $};
		\draw (541,88) node   {$\cdots $};
		\draw (63,44) node   {$\tau \ -1\ \text{clocks}$};
		\draw (322,229) node   {$\tau \ +1\ \text{nodes}$};
		\end{tikzpicture}
		\caption{Structure of the $j$-th block used in Proposition \ref{prop:cycles} to define a conjunctive network with firing memory and maximum delay values $dt_i = \tau$ for every node $i$ that admits attractors with period $k(\tau+1)$. Every circle in the figure represents clock $C(B_j)_{r,l}$ associated to a node $r$ represented by a square. This gadget has $\tau + 1$ nodes and every node has $\tau -1$ clocks.}
		\label{fig:appendixstructblock}
	\end{figure}	
	Finally we consider the graph $G$ as the connected union of $k$-blocks defined connecting every terminal vertex of the $j$-th block to an initial vertex in $(j+1)$-th block and the terminal vertex of $k$-block to the initial vertex of the $1$st-block. With this construction $G$ is a $(\tau+1)$-cycle in which every node is connected to a clock. Using the structure of $G$, we define a global rule $F^{\tau}$ as a conjunctive network with firing memory and maximum delay values $\tau$ which have as underlying interaction graph $G$.  \\	
	Now, we define the dynamics of an attractor with period $\tau + 1$ using $F^{\tau}$.  We define the initial condition $x \in \{0,1\}^{|V|}$ by setting the same initial condition in each block except the first one. For  first block we have the state $0123\cdots\tau$ for the nodes in the path and for $j = 2,3,\cdots,k$ we have the initial states $\tau123\cdots\tau$. For the clocks, every node in a block has one neighbour in a clock and its dynamics is defined as in Proposition \ref{prop:cycles}, so it is associated to an attractor with period $\tau +1$. We are going to write only the state of the node $a$ in every clock and the other states in this subgraph are assumed to be in the initial states so the clock defines an attractor with period $\tau+1$. Assuming this notation, the initial state for the clocks are: $$\underbrace{123\hdots (\tau-1)}_{\text{Node } 0} \text{   } \underbrace{234\hdots\tau}_{\text{Node } 1}  \text{   } \underbrace{  345\hdots \tau 0}_{\text{Node } 2} \hdots   \text{   }   \underbrace{0123\hdots (\tau-2)}_{\text{Node } \tau}.$$ A summary of the initial condition $x$ is shown in Figure \ref{fig:stblocks}.
	\begin{figure}
		\centering	
		\tikzset{every picture/.style={line width=0.75pt}} 
		
		\begin{tikzpicture}[x=0.7pt,y=0.7pt,yscale=-0.5,xscale=0.7]

		\draw    (61,183) -- (215,183) ;

		\draw    (480.5,183) -- (540,183) ;

		\draw    (215,183) -- (377,183) ;

		\draw  [fill={rgb, 255:red, 255; green, 255; blue, 255 }  ,fill opacity=1 ] (180,163) -- (250,163) -- (250,203) -- (180,203) -- cycle ;
		\draw    (180.5,103) -- (180,162) ;

		\draw    (249.5,104) -- (249,163) ;

		\draw  [fill={rgb, 255:red, 255; green, 255; blue, 255 }  ,fill opacity=1 ] (165,103) .. controls (165,94.44) and (171.94,87.5) .. (180.5,87.5) .. controls (189.06,87.5) and (196,94.44) .. (196,103) .. controls (196,111.56) and (189.06,118.5) .. (180.5,118.5) .. controls (171.94,118.5) and (165,111.56) .. (165,103) -- cycle ;
		\draw  [fill={rgb, 255:red, 255; green, 255; blue, 255 }  ,fill opacity=1 ] (234,104) .. controls (234,95.44) and (240.94,88.5) .. (249.5,88.5) .. controls (258.06,88.5) and (265,95.44) .. (265,104) .. controls (265,112.56) and (258.06,119.5) .. (249.5,119.5) .. controls (240.94,119.5) and (234,112.56) .. (234,104) -- cycle ;
		\draw   (264.5,88) .. controls (264.55,83.33) and (262.24,80.98) .. (257.57,80.93) -- (225.57,80.61) .. controls (218.9,80.54) and (215.59,78.17) .. (215.64,73.5) .. controls (215.59,78.17) and (212.24,80.47) .. (205.57,80.4)(208.57,80.43) -- (173.57,80.08) .. controls (168.9,80.03) and (166.55,82.34) .. (166.5,87.01) ;
		\draw  [fill={rgb, 255:red, 255; green, 255; blue, 255 }  ,fill opacity=1 ] (26,163) -- (96,163) -- (96,203) -- (26,203) -- cycle ;
		\draw    (377,183) -- (452.5,183) ;

		\draw  [fill={rgb, 255:red, 255; green, 255; blue, 255 }  ,fill opacity=1 ] (342,163) -- (412,163) -- (412,203) -- (342,203) -- cycle ;
		\draw    (342.5,103) -- (342,162) ;

		\draw    (411.5,104) -- (411,163) ;

		\draw  [fill={rgb, 255:red, 255; green, 255; blue, 255 }  ,fill opacity=1 ] (327,103) .. controls (327,94.44) and (333.94,87.5) .. (342.5,87.5) .. controls (351.06,87.5) and (358,94.44) .. (358,103) .. controls (358,111.56) and (351.06,118.5) .. (342.5,118.5) .. controls (333.94,118.5) and (327,111.56) .. (327,103) -- cycle ;
		\draw  [fill={rgb, 255:red, 255; green, 255; blue, 255 }  ,fill opacity=1 ] (396,104) .. controls (396,95.44) and (402.94,88.5) .. (411.5,88.5) .. controls (420.06,88.5) and (427,95.44) .. (427,104) .. controls (427,112.56) and (420.06,119.5) .. (411.5,119.5) .. controls (402.94,119.5) and (396,112.56) .. (396,104) -- cycle ;
		
		\draw   (426.5,88) .. controls (426.55,83.33) and (424.24,80.98) .. (419.57,80.93) -- (387.57,80.61) .. controls (380.9,80.54) and (377.59,78.17) .. (377.64,73.5) .. controls (377.59,78.17) and (374.24,80.47) .. (367.57,80.4)(370.57,80.43) -- (335.57,80.08) .. controls (330.9,80.03) and (328.55,82.34) .. (328.5,87.01) ;
		
		\draw    (540,183) -- (615.5,183) ;

		\draw  [fill={rgb, 255:red, 255; green, 255; blue, 255 }  ,fill opacity=1 ] (505,163) -- (575,163) -- (575,203) -- (505,203) -- cycle ;
		\draw   (108.5,88) .. controls (108.55,83.33) and (106.24,80.98) .. (101.57,80.93) -- (69.57,80.61) .. controls (62.9,80.54) and (59.59,78.17) .. (59.64,73.5) .. controls (59.59,78.17) and (56.24,80.47) .. (49.57,80.4)(52.57,80.43) -- (17.57,80.08) .. controls (12.9,80.03) and (10.55,82.34) .. (10.5,87.01) ;
		\draw    (26.5,104) -- (26,163) ;

		\draw  [fill={rgb, 255:red, 255; green, 255; blue, 255 }  ,fill opacity=1 ] (11,104) .. controls (11,95.44) and (17.94,88.5) .. (26.5,88.5) .. controls (35.06,88.5) and (42,95.44) .. (42,104) .. controls (42,112.56) and (35.06,119.5) .. (26.5,119.5) .. controls (17.94,119.5) and (11,112.56) .. (11,104) -- cycle ;
		\draw    (95.5,105) -- (95,164) ;

		\draw  [fill={rgb, 255:red, 255; green, 255; blue, 255 }  ,fill opacity=1 ] (80,105) .. controls (80,96.44) and (86.94,89.5) .. (95.5,89.5) .. controls (104.06,89.5) and (111,96.44) .. (111,105) .. controls (111,113.56) and (104.06,120.5) .. (95.5,120.5) .. controls (86.94,120.5) and (80,113.56) .. (80,105) -- cycle ;
		\draw    (505.5,103) -- (505,162) ;

		\draw    (574.5,104) -- (574,163) ;

		\draw  [fill={rgb, 255:red, 255; green, 255; blue, 255 }  ,fill opacity=1 ] (490,103) .. controls (490,94.44) and (496.94,87.5) .. (505.5,87.5) .. controls (514.06,87.5) and (521,94.44) .. (521,103) .. controls (521,111.56) and (514.06,118.5) .. (505.5,118.5) .. controls (496.94,118.5) and (490,111.56) .. (490,103) -- cycle ;
		\draw  [fill={rgb, 255:red, 255; green, 255; blue, 255 }  ,fill opacity=1 ] (559,104) .. controls (559,95.44) and (565.94,88.5) .. (574.5,88.5) .. controls (583.06,88.5) and (590,95.44) .. (590,104) .. controls (590,112.56) and (583.06,119.5) .. (574.5,119.5) .. controls (565.94,119.5) and (559,112.56) .. (559,104) -- cycle ;
		\draw   (588.5,88) .. controls (588.55,83.33) and (586.24,80.98) .. (581.57,80.93) -- (549.57,80.61) .. controls (542.9,80.54) and (539.59,78.17) .. (539.64,73.5) .. controls (539.59,78.17) and (536.24,80.47) .. (529.57,80.4)(532.57,80.43) -- (497.57,80.08) .. controls (492.9,80.03) and (490.55,82.34) .. (490.5,87.01) ;

		\draw (215,100) node   {$\cdots $};
		\draw (377,100) node   {$\cdots $};
		\draw (540,100) node   {$\cdots $};
		\draw (465,181) node   {$\cdots $};
		\draw (61,101) node   {$\cdots $};
		\draw (75+3,48) node   {$\cdots $};
		\draw  [color={rgb, 255:red, 0; green, 0; blue, 0 }  ,draw opacity=1 ][fill={rgb, 255:red, 255; green, 255; blue, 255 }  ,fill opacity=1 ]  (84+8,37) -- (126+12,37) -- (126+12,61) -- (84+8,61) -- cycle  ;
		\draw (105+10,49) node   {$\tau -1$};
		\draw  [color={rgb, 255:red, 0; green, 0; blue, 0 }  ,draw opacity=1 ][fill={rgb, 255:red, 255; green, 255; blue, 255 }  ,fill opacity=1 ]  (4,37) -- (22,37) -- (22,61) -- (4,61) -- cycle  ;
		\draw (13,49) node   {$1$};
		\draw  [color={rgb, 255:red, 0; green, 0; blue, 0 }  ,draw opacity=1 ][fill={rgb, 255:red, 255; green, 255; blue, 255 }  ,fill opacity=1 ]  (25,37) -- (43,37) -- (43,61) -- (25,61) -- cycle  ;
		\draw (34,49) node   {$2$};
		\draw  [color={rgb, 255:red, 0; green, 0; blue, 0 }  ,draw opacity=1 ][fill={rgb, 255:red, 255; green, 255; blue, 255 }  ,fill opacity=1 ]  (46,37) -- (64,37) -- (64,61) -- (46,61) -- cycle  ;
		\draw (55,49) node   {$3$};
		\draw (231,48) node   {$\cdots $};
		\draw  [color={rgb, 255:red, 0; green, 0; blue, 0 }  ,draw opacity=1 ][fill={rgb, 255:red, 255; green, 255; blue, 255 }  ,fill opacity=1 ]  (252,37) -- (270,37) -- (270,61) -- (252,61) -- cycle  ;
		\draw (261,49) node   {$\tau $};
		\draw  [color={rgb, 255:red, 0; green, 0; blue, 0 }  ,draw opacity=1 ][fill={rgb, 255:red, 255; green, 255; blue, 255 }  ,fill opacity=1 ]  (160,37) -- (178,37) -- (178,61) -- (160,61) -- cycle  ;
		\draw (169,49) node   {$2$};
		\draw  [color={rgb, 255:red, 0; green, 0; blue, 0 }  ,draw opacity=1 ][fill={rgb, 255:red, 255; green, 255; blue, 255 }  ,fill opacity=1 ]  (181,37) -- (199,37) -- (199,61) -- (181,61) -- cycle  ;
		\draw (190,49) node   {$3$};
		\draw  [color={rgb, 255:red, 0; green, 0; blue, 0 }  ,draw opacity=1 ][fill={rgb, 255:red, 255; green, 255; blue, 255 }  ,fill opacity=1 ]  (202,37) -- (220,37) -- (220,61) -- (202,61) -- cycle  ;
		\draw (211,49) node   {$4$};
		\draw (384,48) node   {$\cdots $};
		\draw  [color={rgb, 255:red, 0; green, 0; blue, 0 }  ,draw opacity=1 ][fill={rgb, 255:red, 255; green, 255; blue, 255 }  ,fill opacity=1 ]  (396,37) -- (414,37) -- (414,61) -- (396,61) -- cycle  ;
		\draw (405,49) node   {$\tau $};
		\draw  [color={rgb, 255:red, 0; green, 0; blue, 0 }  ,draw opacity=1 ][fill={rgb, 255:red, 255; green, 255; blue, 255 }  ,fill opacity=1 ]  (313,37) -- (331,37) -- (331,61) -- (313,61) -- cycle  ;
		\draw (322,49) node   {$3$};
		\draw  [color={rgb, 255:red, 0; green, 0; blue, 0 }  ,draw opacity=1 ][fill={rgb, 255:red, 255; green, 255; blue, 255 }  ,fill opacity=1 ]  (334,37) -- (352,37) -- (352,61) -- (334,61) -- cycle  ;
		\draw (343,49) node   {$4$};
		\draw  [color={rgb, 255:red, 0; green, 0; blue, 0 }  ,draw opacity=1 ][fill={rgb, 255:red, 255; green, 255; blue, 255 }  ,fill opacity=1 ]  (355,37) -- (373,37) -- (373,61) -- (355,61) -- cycle  ;
		\draw (364,49) node   {$5$};
		\draw (553+3,48) node   {$\cdots $};
		\draw  [color={rgb, 255:red, 0; green, 0; blue, 0 }  ,draw opacity=1 ][fill={rgb, 255:red, 255; green, 255; blue, 255 }  ,fill opacity=1 ]  (562+8,37) -- (604+12,37) -- (604+12,61) -- (562+8,61) -- cycle  ;
		\draw (583+10,49) node   {$\tau -2$};
		\draw  [color={rgb, 255:red, 0; green, 0; blue, 0 }  ,draw opacity=1 ][fill={rgb, 255:red, 255; green, 255; blue, 255 }  ,fill opacity=1 ]  (482,37) -- (500,37) -- (500,61) -- (482,61) -- cycle  ;
		\draw (491,49) node   {$0$};
		\draw  [color={rgb, 255:red, 0; green, 0; blue, 0 }  ,draw opacity=1 ][fill={rgb, 255:red, 255; green, 255; blue, 255 }  ,fill opacity=1 ]  (503,37) -- (521,37) -- (521,61) -- (503,61) -- cycle  ;
		\draw (512,49) node   {$1$};
		\draw  [color={rgb, 255:red, 0; green, 0; blue, 0 }  ,draw opacity=1 ][fill={rgb, 255:red, 255; green, 255; blue, 255 }  ,fill opacity=1 ]  (524,37) -- (542,37) -- (542,61) -- (524,61) -- cycle  ;
		\draw (533,49) node   {$2$};
		\draw  [color={rgb, 255:red, 0; green, 0; blue, 0 }  ,draw opacity=1 ][fill={rgb, 255:red, 255; green, 255; blue, 255 }  ,fill opacity=1 ]  (417,37) -- (435,37) -- (435,61) -- (417,61) -- cycle  ;
		\draw (426,49) node   {$0$};
		\draw (60,182) node   {$\tau $};
		\draw (215,182) node   {$1$};
		\draw (375,182) node   {$2$};
		\draw (541,182) node   {$\tau $};	
		\end{tikzpicture}	
		\caption{Initial condition for the $j$-th block, $j\geq 2$ used in Proposition \ref{prop:cycles} to define an attractor with period $k(\tau+1).$ For the first block we just define the state of the first node to $0$ instead of $\tau$.}
		\label{fig:stblocks}		
	\end{figure}
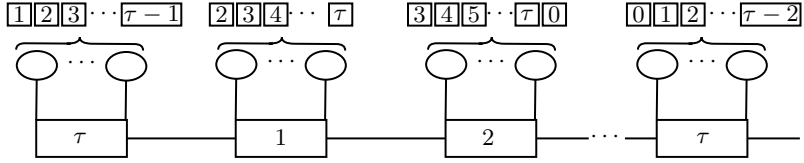	
	We remark that the dynamics of nodes labelled by $a$ is defined by an attractor with period $\tau + 1$ and it is independent of the dynamics of its neighbours in the path. Then, to prove that the global period of the network is  $k(\tau+1)$ it suffices to show that for the first block the next state will be $\tau0123\hdots\tau-1$. If we have that in $\tau+1$ steps the second block will have the state $0123\cdots\tau$ and so in $k(\tau +1)$ steps the network will return to the initial condition. 
	
	In fact, we have that every node in the path except for the first and second node has some neighbour $a$ in a clock and in state $0$. But as it is shown in Figure \ref{fig:ktauc}.  node $1$ has node $0$ as a neighbour and node $0$ is initial setted to $0$ so in the next step, every node in the path is going descend to the previous state except for the first one that will be setted to $\tau$ thus, the next state of the nodes in the path will be $\tau0123\hdots\tau-1$ as desired. Besides, clocks will be updated accordingly: $$\underbrace{012\hdots (\tau-2)}_{\text{Node } 0} \text{   } \underbrace{123\hdots\tau-1}_{\text{Node } 1}  \text{   } \underbrace{  234\hdots \tau-1 \tau}_{\text{Node } 3} \hdots   \text{   }   \underbrace{\tau012\hdots (\tau-2)}_{\text{Node } \tau}$$
	
	We conclude that $F$ admits attractors with lenght $k(\tau +1)$.
	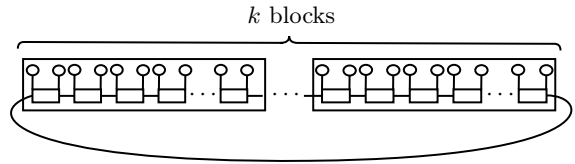
\begin{figure}
		
		\tikzset{every picture/.style={line width=0.75pt}} 
		\centering
		\begin{tikzpicture}[x=0.6pt,y=0.6pt,yscale=-0.6,xscale=0.55]
		
		\draw    (339.39,141.73) -- (372.35,141.73) ;


		\draw  [fill={rgb, 255:red, 255; green, 255; blue, 255 }  ,fill opacity=1 ] (585.62,134.89) -- (617.48,134.89) -- (617.48,148.56) -- (585.62,148.56) -- cycle ;
		\draw    (585.85,114.39) -- (585.62,134.55) ;

		\draw    (617.25,114.73) -- (617.02,134.89) ;

		\draw  [fill={rgb, 255:red, 255; green, 255; blue, 255 }  ,fill opacity=1 ] (610.19,114.73) .. controls (610.19,111.81) and (613.35,109.43) .. (617.25,109.43) .. controls (621.14,109.43) and (624.3,111.81) .. (624.3,114.73) .. controls (624.3,117.66) and (621.14,120.03) .. (617.25,120.03) .. controls (613.35,120.03) and (610.19,117.66) .. (610.19,114.73) -- cycle ;
		

		\draw    (426.14,141.73) -- (460.5,141.73) ;

		\draw    (376.08,141.73) -- (410.44,141.73) ;

		\draw  [fill={rgb, 255:red, 255; green, 255; blue, 255 }  ,fill opacity=1 ] (360.15,134.89) -- (392.01,134.89) -- (392.01,148.56) -- (360.15,148.56) -- cycle ;
		\draw  [fill={rgb, 255:red, 255; green, 255; blue, 255 }  ,fill opacity=1 ] (410.21,134.89) -- (442.07,134.89) -- (442.07,148.56) -- (410.21,148.56) -- cycle ;
		\draw    (476.66,141.73) -- (511.02,141.73) ;

		\draw  [fill={rgb, 255:red, 255; green, 255; blue, 255 }  ,fill opacity=1 ] (460.73,134.89) -- (492.59,134.89) -- (492.59,148.56) -- (460.73,148.56) -- cycle ;
		\draw  [fill={rgb, 255:red, 255; green, 255; blue, 255 }  ,fill opacity=1 ] (510.79,134.89) -- (542.65,134.89) -- (542.65,148.56) -- (510.79,148.56) -- cycle ;
		\draw    (360.38,114.73) -- (360.15,134.89) ;

		\draw  [fill={rgb, 255:red, 255; green, 255; blue, 255 }  ,fill opacity=1 ] (353.33,114.73) .. controls (353.33,111.81) and (356.49,109.43) .. (360.38,109.43) .. controls (364.28,109.43) and (367.44,111.81) .. (367.44,114.73) .. controls (367.44,117.66) and (364.28,120.03) .. (360.38,120.03) .. controls (356.49,120.03) and (353.33,117.66) .. (353.33,114.73) -- cycle ;
		\draw    (391.78,115.07) -- (391.56,135.24) ;

		\draw  [fill={rgb, 255:red, 255; green, 255; blue, 255 }  ,fill opacity=1 ] (384.73,115.07) .. controls (384.73,112.15) and (387.89,109.78) .. (391.78,109.78) .. controls (395.68,109.78) and (398.84,112.15) .. (398.84,115.07) .. controls (398.84,118) and (395.68,120.37) .. (391.78,120.37) .. controls (387.89,120.37) and (384.73,118) .. (384.73,115.07) -- cycle ;
		\draw    (410.44,114.39) -- (410.21,134.55) ;

		\draw    (441.84,114.73) -- (441.62,134.89) ;

		\draw    (460.96,114.39) -- (460.73,134.55) ;

		\draw    (492.36,114.73) -- (492.13,134.89) ;

		\draw    (511.47,114.39) -- (511.25,134.55) ;

		\draw    (542.88,114.73) -- (542.65,134.89) ;

		\draw  [fill={rgb, 255:red, 255; green, 255; blue, 255 }  ,fill opacity=1 ] (403.39,114.39) .. controls (403.39,111.46) and (406.55,109.09) .. (410.44,109.09) .. controls (414.34,109.09) and (417.5,111.46) .. (417.5,114.39) .. controls (417.5,117.31) and (414.34,119.69) .. (410.44,119.69) .. controls (406.55,119.69) and (403.39,117.31) .. (403.39,114.39) -- cycle ;
		\draw  [fill={rgb, 255:red, 255; green, 255; blue, 255 }  ,fill opacity=1 ] (434.79,114.73) .. controls (434.79,111.81) and (437.95,109.43) .. (441.84,109.43) .. controls (445.74,109.43) and (448.9,111.81) .. (448.9,114.73) .. controls (448.9,117.66) and (445.74,120.03) .. (441.84,120.03) .. controls (437.95,120.03) and (434.79,117.66) .. (434.79,114.73) -- cycle ;
		\draw  [fill={rgb, 255:red, 255; green, 255; blue, 255 }  ,fill opacity=1 ] (453.9,114.39) .. controls (453.9,111.46) and (457.06,109.09) .. (460.96,109.09) .. controls (464.85,109.09) and (468.01,111.46) .. (468.01,114.39) .. controls (468.01,117.31) and (464.85,119.69) .. (460.96,119.69) .. controls (457.06,119.69) and (453.9,117.31) .. (453.9,114.39) -- cycle ;
		\draw  [fill={rgb, 255:red, 255; green, 255; blue, 255 }  ,fill opacity=1 ] (485.31,114.73) .. controls (485.31,111.81) and (488.46,109.43) .. (492.36,109.43) .. controls (496.26,109.43) and (499.41,111.81) .. (499.41,114.73) .. controls (499.41,117.66) and (496.26,120.03) .. (492.36,120.03) .. controls (488.46,120.03) and (485.31,117.66) .. (485.31,114.73) -- cycle ;
		\draw  [fill={rgb, 255:red, 255; green, 255; blue, 255 }  ,fill opacity=1 ] (504.42,114.39) .. controls (504.42,111.46) and (507.58,109.09) .. (511.47,109.09) .. controls (515.37,109.09) and (518.53,111.46) .. (518.53,114.39) .. controls (518.53,117.31) and (515.37,119.69) .. (511.47,119.69) .. controls (507.58,119.69) and (504.42,117.31) .. (504.42,114.39) -- cycle ;
		\draw  [fill={rgb, 255:red, 255; green, 255; blue, 255 }  ,fill opacity=1 ] (535.82,114.73) .. controls (535.82,111.81) and (538.98,109.43) .. (542.88,109.43) .. controls (546.77,109.43) and (549.93,111.81) .. (549.93,114.73) .. controls (549.93,117.66) and (546.77,120.03) .. (542.88,120.03) .. controls (538.98,120.03) and (535.82,117.66) .. (535.82,114.73) -- cycle ;
		\draw  [fill={rgb, 255:red, 255; green, 255; blue, 255 }  ,fill opacity=1 ] (577.86,114.39) .. controls (577.86,111.46) and (581.02,109.09) .. (584.91,109.09) .. controls (588.81,109.09) and (591.97,111.46) .. (591.97,114.39) .. controls (591.97,117.31) and (588.81,119.69) .. (584.91,119.69) .. controls (581.02,119.69) and (577.86,117.31) .. (577.86,114.39) -- cycle ;
		
		\draw   (350.06,103) -- (627.5,103) -- (627.5,157) -- (350.06,157) -- cycle ;
		

		\draw    (259.18,141.73) -- (292.14,141.73) ;

		\draw  [fill={rgb, 255:red, 255; green, 255; blue, 255 }  ,fill opacity=1 ] (243.9,134.89) -- (274.46,134.89) -- (274.46,148.56) -- (243.9,148.56) -- cycle ;
		\draw    (244.12,114.39) -- (243.9,134.55) ;

		\draw  [fill={rgb, 255:red, 255; green, 255; blue, 255 }  ,fill opacity=1 ] (237.35,114.39) .. controls (237.35,111.46) and (240.38,109.09) .. (244.12,109.09) .. controls (247.86,109.09) and (250.89,111.46) .. (250.89,114.39) .. controls (250.89,117.31) and (247.86,119.69) .. (244.12,119.69) .. controls (240.38,119.69) and (237.35,117.31) .. (237.35,114.39) -- cycle ;
		
		\draw    (274.24,114.73) -- (274.02,134.89) ;

		\draw  [fill={rgb, 255:red, 255; green, 255; blue, 255 }  ,fill opacity=1 ] (267.47,114.73) .. controls (267.47,111.81) and (270.5,109.43) .. (274.24,109.43) .. controls (277.98,109.43) and (281.01,111.81) .. (281.01,114.73) .. controls (281.01,117.66) and (277.98,120.03) .. (274.24,120.03) .. controls (270.5,120.03) and (267.47,117.66) .. (267.47,114.73) -- cycle ;
		

		\draw    (91.55,141.73) -- (124.51,141.73) ;

		\draw    (43.53,141.73) -- (76.49,141.73) ;

		\draw  [fill={rgb, 255:red, 255; green, 255; blue, 255 }  ,fill opacity=1 ] (28.25,134.89) -- (58.81,134.89) -- (58.81,148.56) -- (28.25,148.56) -- cycle ;
		\draw  [fill={rgb, 255:red, 255; green, 255; blue, 255 }  ,fill opacity=1 ] (76.27,134.89) -- (106.83,134.89) -- (106.83,148.56) -- (76.27,148.56) -- cycle ;
		\draw    (140.01,141.73) -- (172.97,141.73) ;

		\draw  [fill={rgb, 255:red, 255; green, 255; blue, 255 }  ,fill opacity=1 ] (124.73,134.89) -- (155.29,134.89) -- (155.29,148.56) -- (124.73,148.56) -- cycle ;
		\draw  [fill={rgb, 255:red, 255; green, 255; blue, 255 }  ,fill opacity=1 ] (172.75,134.89) -- (203.31,134.89) -- (203.31,148.56) -- (172.75,148.56) -- cycle ;
		\draw    (28.47,114.73) -- (28.25,134.89) ;

		\draw  [fill={rgb, 255:red, 255; green, 255; blue, 255 }  ,fill opacity=1 ] (21.7,114.73) .. controls (21.7,111.81) and (24.73,109.43) .. (28.47,109.43) .. controls (32.21,109.43) and (35.24,111.81) .. (35.24,114.73) .. controls (35.24,117.66) and (32.21,120.03) .. (28.47,120.03) .. controls (24.73,120.03) and (21.7,117.66) .. (21.7,114.73) -- cycle ;
		\draw    (58.59,115.07) -- (58.37,135.24) ;

		\draw  [fill={rgb, 255:red, 255; green, 255; blue, 255 }  ,fill opacity=1 ] (51.83,115.07) .. controls (51.83,112.15) and (54.86,109.78) .. (58.59,109.78) .. controls (62.33,109.78) and (65.36,112.15) .. (65.36,115.07) .. controls (65.36,118) and (62.33,120.37) .. (58.59,120.37) .. controls (54.86,120.37) and (51.83,118) .. (51.83,115.07) -- cycle ;
		\draw    (76.49,114.39) -- (76.27,134.55) ;

		\draw    (106.61,114.73) -- (106.4,134.89) ;

		\draw    (124.95,114.39) -- (124.73,134.55) ;

		\draw    (155.07,114.73) -- (154.85,134.89) ;

		\draw    (173.41,114.39) -- (173.19,134.55) ;

		\draw    (203.53,114.73) -- (203.31,134.89) ;

		\draw  [fill={rgb, 255:red, 255; green, 255; blue, 255 }  ,fill opacity=1 ] (69.72,114.39) .. controls (69.72,111.46) and (72.75,109.09) .. (76.49,109.09) .. controls (80.23,109.09) and (83.26,111.46) .. (83.26,114.39) .. controls (83.26,117.31) and (80.23,119.69) .. (76.49,119.69) .. controls (72.75,119.69) and (69.72,117.31) .. (69.72,114.39) -- cycle ;
		\draw  [fill={rgb, 255:red, 255; green, 255; blue, 255 }  ,fill opacity=1 ] (99.85,114.73) .. controls (99.85,111.81) and (102.88,109.43) .. (106.61,109.43) .. controls (110.35,109.43) and (113.38,111.81) .. (113.38,114.73) .. controls (113.38,117.66) and (110.35,120.03) .. (106.61,120.03) .. controls (102.88,120.03) and (99.85,117.66) .. (99.85,114.73) -- cycle ;
		\draw  [fill={rgb, 255:red, 255; green, 255; blue, 255 }  ,fill opacity=1 ] (118.18,114.39) .. controls (118.18,111.46) and (121.21,109.09) .. (124.95,109.09) .. controls (128.69,109.09) and (131.72,111.46) .. (131.72,114.39) .. controls (131.72,117.31) and (128.69,119.69) .. (124.95,119.69) .. controls (121.21,119.69) and (118.18,117.31) .. (118.18,114.39) -- cycle ;
		\draw  [fill={rgb, 255:red, 255; green, 255; blue, 255 }  ,fill opacity=1 ] (148.3,114.73) .. controls (148.3,111.81) and (151.33,109.43) .. (155.07,109.43) .. controls (158.81,109.43) and (161.84,111.81) .. (161.84,114.73) .. controls (161.84,117.66) and (158.81,120.03) .. (155.07,120.03) .. controls (151.33,120.03) and (148.3,117.66) .. (148.3,114.73) -- cycle ;
		\draw  [fill={rgb, 255:red, 255; green, 255; blue, 255 }  ,fill opacity=1 ] (166.64,114.39) .. controls (166.64,111.46) and (169.67,109.09) .. (173.41,109.09) .. controls (177.14,109.09) and (180.17,111.46) .. (180.17,114.39) .. controls (180.17,117.31) and (177.14,119.69) .. (173.41,119.69) .. controls (169.67,119.69) and (166.64,117.31) .. (166.64,114.39) -- cycle ;
		\draw  [fill={rgb, 255:red, 255; green, 255; blue, 255 }  ,fill opacity=1 ] (196.76,114.73) .. controls (196.76,111.81) and (199.79,109.43) .. (203.53,109.43) .. controls (207.27,109.43) and (210.3,111.81) .. (210.3,114.73) .. controls (210.3,117.66) and (207.27,120.03) .. (203.53,120.03) .. controls (199.79,120.03) and (196.76,117.66) .. (196.76,114.73) -- cycle ;
		
		\draw   (17.5,103) -- (294.94,103) -- (294.94,157) -- (17.5,157) -- cycle ;
		
		\draw    (28.25,141.46) .. controls (-31.75,158) and (-1.5,213) .. (345.5,210) .. controls (692.5,207) and (666.25,140.76) .. (617.75,141.26) ;

		\draw   (633.5,94) .. controls (633.51,89.33) and (631.19,86.99) .. (626.52,86.98) -- (332.52,86.05) .. controls (325.85,86.03) and (322.53,83.69) .. (322.54,79.02) .. controls (322.53,83.69) and (319.19,86.01) .. (312.52,85.99)(315.52,86) -- (18.52,85.07) .. controls (13.85,85.06) and (11.51,87.38) .. (11.5,92.05) ;

		
		
		
		
		\draw (223.64,140.2) node   {$\cdots $};
		\draw (564.61,140.2) node   {$\cdots $};
		
		\draw (318.76,140.13) node   {$\cdots $};
		\draw (325,56) node   {$k\ \text{blocks}$};

		\end{tikzpicture}

		\caption{Interaction graph $G$ associated to a conjunctive network with firing memory and maximum delay $dt_i = \tau$ for all $i \in V(G)$, defined in Proposition \ref{prop:cycles} that admits attractors with length $k(\tau +1)$. }
		\label{fig:ktauc}
	\end{figure}
\end{proof}

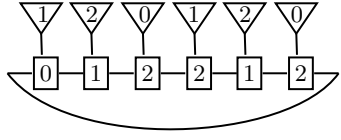
\begin{figure}
	\centering
	\tikzset{every picture/.style={line width=0.75pt}} 
	
	\begin{tikzpicture}[x=0.5pt,y=0.5pt,yscale=-1,xscale=1]
	
	\draw    (239,60) -- (257.5,60) ;

	\draw    (226,46) -- (225.5,30) ;

	\draw    (277,60) -- (295.5,60) ;

	\draw    (264,46) -- (263.5,30) ;

	\draw    (316,60) -- (334.5,60) ;

	\draw    (303,46) -- (302.5,30) ;

	\draw    (355,60) -- (373.5,60) ;

	\draw    (342,46) -- (341.5,30) ;

	\draw    (393,60) -- (411.5,60) ;

	\draw    (380,46) -- (379.5,30) ;

	\draw    (432,60) -- (450.5,60) ;

	\draw    (419,46) -- (418.5,30) ;

	\draw    (200,60) -- (218.5,60) ;

	\draw    (200,60) .. controls (245.5,120) and (412.5,113) .. (450.5,60) ;

	\draw   (225.5,30) -- (209.88,7) -- (241.13,7) -- cycle ;
	\draw   (263.5,30) -- (247.88,7) -- (279.13,7) -- cycle ;
	\draw   (302.5,30) -- (286.88,7) -- (318.13,7) -- cycle ;
	\draw   (341.5,30) -- (325.88,7) -- (357.13,7) -- cycle ;
	\draw   (379.5,30) -- (363.88,7) -- (395.13,7) -- cycle ;
	\draw   (418.5,30) -- (402.88,7) -- (434.13,7) -- cycle ;

	\draw    (220,48) -- (238,48) -- (238,72) -- (220,72) -- cycle  ;
	\draw (229,60) node   {$0$};
	\draw    (258,48) -- (276,48) -- (276,72) -- (258,72) -- cycle  ;
	\draw (267,60) node   {$1$};
	\draw    (297,48) -- (315,48) -- (315,72) -- (297,72) -- cycle  ;
	\draw (306,60) node   {$2$};
	\draw    (336,48) -- (354,48) -- (354,72) -- (336,72) -- cycle  ;
	\draw (345,60) node   {$2$};
	\draw    (374,48) -- (392,48) -- (392,72) -- (374,72) -- cycle  ;
	\draw (383,60) node   {$1$};
	\draw    (413,48) -- (431,48) -- (431,72) -- (413,72) -- cycle  ;
	\draw (422,60) node   {$2$};
	\draw (227,15) node   {$1$};
	\draw (264,15) node   {$2$};
	\draw (304,15) node   {$0$};
	\draw (341,15) node   {$1$};
	\draw (380,15) node   {$2$};
	\draw (419,15) node   {$0$};
	\end{tikzpicture}	
	\caption{An interaction graph $G$ associated to a conjunctive network with firing memory and maximum delay values $dt_i = 2$ for all $i \in V(G)$ that admits attractors with period $6$. Every triangle represents a conjunctive network with firing memory that admits attractors with period $3$ and states in the triangles represent the states of the nodes that are connected to the nodes in the path. }
	\label{fig:examplecycle}
\end{figure}
\begin{teo}
	Let $\tau \geq 2$.  There exists a conjunctive network with firing memory and maximum delay  $dt_i = \tau$ in every node $i$ which admits attractors with non polynomial period.
	\label{Appendix:nonpolycycle}
\end{teo}
\begin{proof}
	Let $m \geq 2$ and a collection of prime numbers $p_1, \hdots,  p_l$ where $l = \pi(m)$ as in Theorem \ref{teononpolydif}. As a consequence of Proposition \ref{propcycleshom} there exist functions $F_{p_i}$ and associated graphs $C_{p_i}, \text{ } i = 1,2,\cdots, l$  which admit attractors with period $p_i (\tau+1)$. As same as we did to prove Theorem \ref{teononpolydif} we define a global function $F$ with delay $\tau$ in every coordinate which associated interaction graph $G = (V,E)$ is given by the connected union of the graphs $C_{p_i}$. In this case, we connect every of these components by adding an edge between a node labelled by $a$ associated to the first vertex in the path of the first block of $C_{p_i}$ to another $a$ labelled vertex associated to the second vertex in the path of the first block of $C_{p_{i+1}}$ as it is shown in Figure \ref{fig:nonpolys}. Initial condition $x$ is defined as every node in $C_{p_{i+1}}$ is in an attractor with period $(\tau + 1)p_i$ by using the same initial condition given in the proof of Proposition \ref{propcycleshom}.  We remark that there no connected vertices with state $0$ in every iteration because of how we defined the connection between components. Again, as same as in Theorem \ref{teononpolydif} we have that the global period of the network $T$ satisfy that: 
	\begin{equation*}
	T \geq(\tau + 1)  \prod  \limits_{i=1}^{\pi(m)} p_i.
	\end{equation*}
	And also, we have that:	
	\begin{equation}
	|V| = (\tau+1)\tau \sum \limits_{i=1}^{\pi(m)} p_i.
	\label{vertexbound}
	\end{equation}	
	It is not difficult to see that applying the same technique that we used in the proof of Theorem \ref{teononpolydif} we can conclude:	
	\begin{equation*}
	T \geq \exp(\Omega(\sqrt{|V|\log(|V|)}).
	\end{equation*}
	Thus,  $F$ has attractors with non polynomial period.
\end{proof}
\begin{figure}
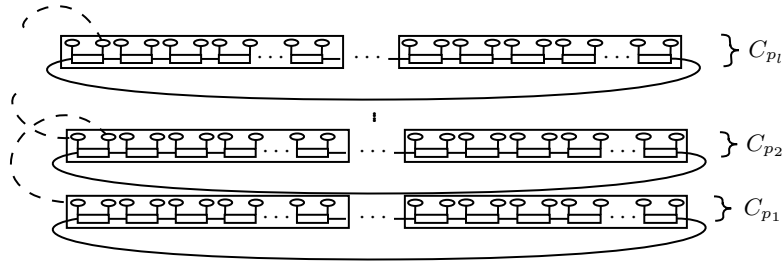

	\tikzset{every picture/.style={line width=0.75pt}} 
	\centering	

	\caption{Interaction graph  $G$ associated to a conjunctive network with firing memory and maximum delay values $dt_i = \tau$ for every node $i \in V(G)$ that admits attractors with non polynomial period. Every component defines a local dynamics with period $(\tau + 1) p_i$. Initial condition is defined verifying that there are no connected nodes in $0$. Global period of the network is given by the product of prime numbers $p_i$. }
	\label{fig:nonpolys}
\end{figure}

\newpage
\subsection{Supplementary figures: dynamics of logic gates gadgets}

\begin{figure}[h!]
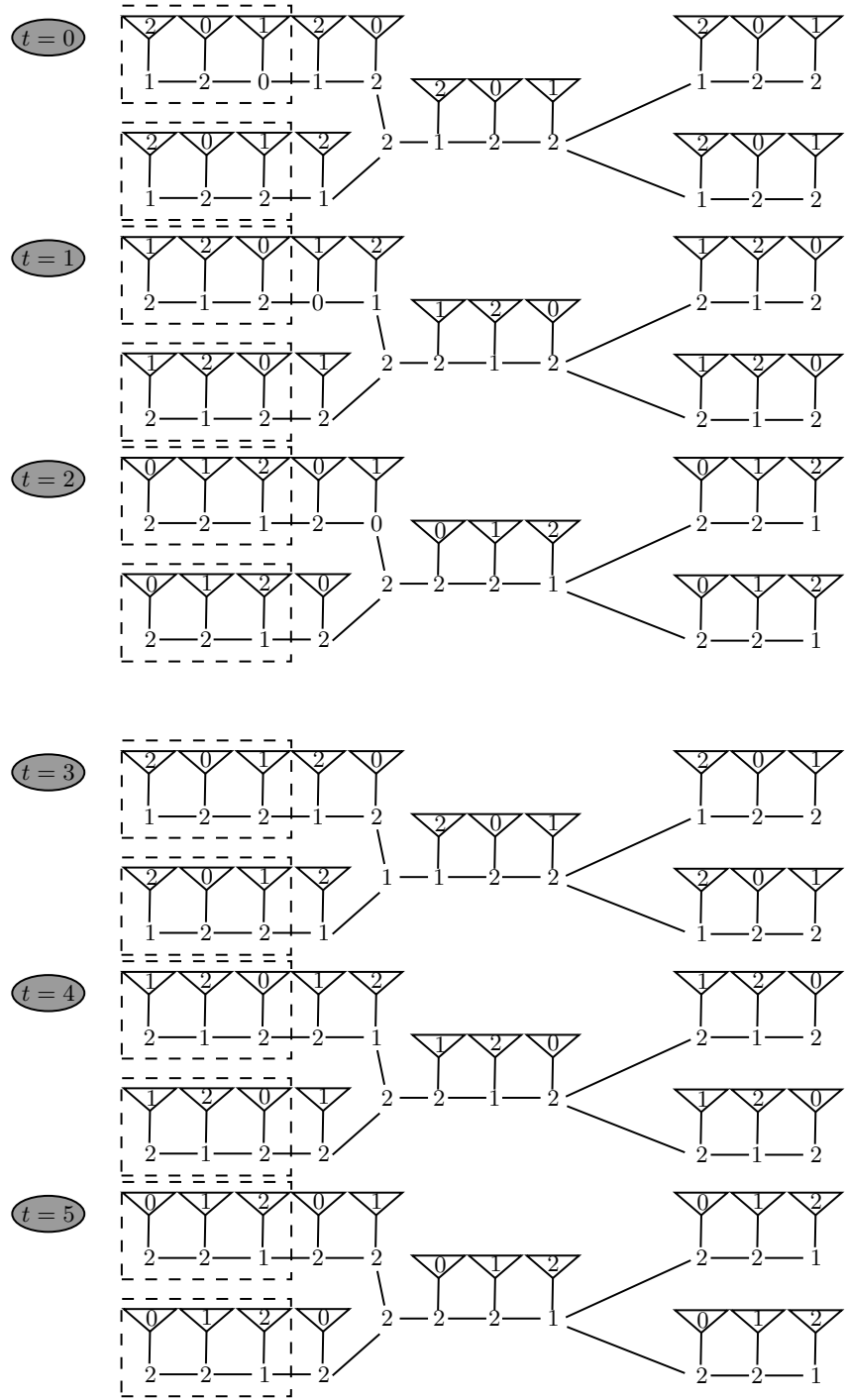

	\centering

		\centering

		\tikzset{every picture/.style={line width=0.75pt}} 
		

		
	\caption{Iterations of the AND gate gadget. A $0$ and a $1$ are computed as inputs. After seven steps the information is transmitted and initial condition is recovered.}
	\label{fig:iterAND}
\end{figure}

\begin{figure}[h!]
	\tikzset{every picture/.style={line width=0.75pt}} 
\centering
	%
	\caption{Iterations of the OR gate gadget. A $0$ and a $1$ are computed as inputs. After three steps the information is transmitted and the initial condition is recovered.}
	\label{fig:iterOR}
\end{figure}

\end{document}